%% file: main.tex
\pdfoutput=1
\documentclass[11pt,acmsmall,screen,nonacm]{acmart}
\usepackage{natbib}
\usepackage{dsfont}
\usepackage{xspace}
\makeatletter
\let\@authorsaddresses\@empty
\makeatother
\usepackage[linesnumbered,lined,ruled,noend]{algorithm2e}
\DontPrintSemicolon
\SetKwFor{RepTimes}{Repeat}{~times:}{end}
\SetKwProg{function}{Function}{:}{}
\SetKwProg{parfor}{parallel for}{\:do}{end}
\usepackage{array} 
\usepackage{graphicx}

\usepackage{hyperref}

\usepackage{pgfplots}
\pgfplotsset{compat=1.17}
\usepackage{mathtools}
\usepackage[shortlabels, inline]{enumitem}
\usepackage{color}
\usepackage{tikz}

\usepackage{amsmath}
\usepackage{amsthm}
\usepackage{graphicx,epsfig}
\usepackage{caption}
\usepackage[noend]{algpseudocode}

\usepackage{hyperref}
\usepackage{cleveref}
\hypersetup{
    colorlinks=true,
    linktoc=all,
    allcolors=blue,
}

\SetArgSty{textrm}
\SetKwProg{Fn}{Function}{}{}

\DeclareMathOperator{\poly}{polylog}

\newcommand{\E}{\mathbf{E}}

\newcommand{\mc}[1]{\mathcal{#1}}

\newcommand{\mrkd}{\mathrm{Mrkd}}
\newcommand{\upmrkd}{\mathrm{Mrkd}_U}

\newcommand{\lmrkd}{\mathrm{Mrkd}_L}

\newcommand{\umrkd}{\mathrm{Unmrkd}}

\newcommand{\blank}{\mathrm{Blank}}

\newcommand{\nbr}[2]{\mathsf{NbrHT}[#1,#2]}
\newcommand{\belowHT}[1]{\mathsf{BelowHT}[#1]}

\numberwithin{equation}{section}

\newtheorem{theorem}{Theorem}[section]
\newtheorem{lemma}{Lemma}[section]

\newtheorem{remark}{Remark}[section]

\newtheorem{claim}{Claim}[lemma]

\newtheorem{condition}{Condition}[section]

\renewcommand{\E}{\mathbf{E}}

\newcommand{\whp}{\emph{whp}\xspace}

\newcommand{\ind}{\mathbf{1}}

\newcommand{\Ind}[1]{\ind\left[#1\right]}

\newlist{caselist}{enumerate}{1}
\setlist[caselist]{
  label=\textbf{Case~\arabic*:},
}

\begin{document}

\title{Parallel Batch Dynamic Vertex Coloring in $O(\log \Delta)$ Amortized Update Time}

\author{Chase Hutton}
\affiliation{
    \institution{University of Maryland}
    \country{United States}
}

\author{Adam Melrod}
\affiliation{
    \institution{Harvard University}
    \country{United States}
}

\begin{abstract}
We present the first parallel batch-dynamic algorithm for maintaining a proper $(\Delta + 1)$-vertex coloring. Our approach builds on a new sequential dynamic algorithm inspired by the work of Bhattacharya et al.\ (SODA'18). The resulting randomized algorithm achieves $O(\log \Delta)$ expected amortized update time and, for any batch of $b$ updates, has parallel span $O(\operatorname{polylog} b + \operatorname{polylog} n)$ with high probability.
\end{abstract}

\maketitle

\input{sections/introduction}
\input{sections/technical-overview}
\input{sections/prelims}

\input{sections/static-coloring}
\input{sections/implementation}

\input{sections/algorithm1}

\bibliographystyle{ACM-Reference-Format}
\bibliography{strings,refs}
\end{document}

%% file: sections/introduction.tex
\section{Introduction}\label{sec:introduction}
Vertex coloring is a fundamental problem in computer science with many applications.
For an undirected graph $G = (V,E)$ and an integer parameter $c > 0$, a proper $c$-vertex coloring of $G$ assigns to every vertex $u \in V$ a color from a palette $\mc C = \{0,\dots,c-1\}$ such that no edge has both endpoints with the same color.
Typically, the goal is to find a proper coloring using as few colors as possible.
Unfortunately, even approximating the smallest such $c$ is hard: for any constant $\epsilon > 0$, there is no polynomial-time algorithm with approximation factor $n^{1-\epsilon}$ unless $P \neq NP$~\cite{hardness-coloring-1}.
On the positive side, a textbook greedy algorithm computes a $(\Delta+1)$-coloring in $O(|E|)$ time, where $\Delta$ is the maximum degree of $G$.
The algorithm maintains a palette $\mc C_u$ for each vertex $u$, initially $\mc C_u = \{0,\dots,\Delta\}$.
While there exists an uncolored vertex $u$, the algorithm scans $\mc C_u$ to find a color not used by any neighbor of $u$, assigns this color to $u$, and then removes that color from the palettes of $u$’s neighbors.
Such a color always exists because $u$ has at most $\Delta$ neighbors and $|\mc C_u| = \Delta+1$.

In this paper, we study the problem of maintaining a proper $(\Delta+1)$-coloring in the \emph{dynamic} setting.
In the classical (sequential) dynamic model, the graph undergoes a sequence of \emph{updates}, where each update inserts or deletes a single edge.
The goal is to design a dynamic algorithm that maintains a proper $(\Delta+1)$-coloring $\chi$ as the graph changes.
The time taken by the algorithm to process an update is called its \emph{update time}.
It is straightforward to design a dynamic algorithm with $O(\Delta)$ worst-case update time: the greedy $(\Delta+1)$-coloring algorithm described above can be adapted to recolor locally after each update.
A more interesting question is whether one can achieve \emph{sublinear} update time, i.e., $o(\Delta)$.
A sequence of works~\cite{Dynamic-Coloring-18,Bhattacharya-Coloring-Constant,Henzinger2019ConstantTimeD} answered this question in the affirmative, with the best bound due to Bhattacharya et al.~\cite{Bhattacharya-Coloring-Constant}, stated below.
We call a dynamic graph \emph{$\Delta$-bounded} if at all times its maximum degree is at most $\Delta$.

\begin{theorem}[\cite{Bhattacharya-Coloring-Constant}]
    There is a randomized data structure for maintaining a $(\Delta+1)$-coloring in a $\Delta$-bounded graph that, given any sequence of $t$ updates, takes total time $O(n \log n + n\Delta + t)$ in expectation and with high probability.
    The space usage is $O(n \log n + n\Delta + m)$, where $m$ is the maximum number of edges present at any time.
\end{theorem}

Motivated by these sequential results, we ask whether similar guarantees can be obtained in the \emph{parallel batch-dynamic} setting.
Here, updates arrive in \emph{batches}, and the goal is to maintain a proper $(\Delta+1)$-coloring while minimizing both the total \emph{work} (the number of operations performed) and the \emph{span} (the length of the longest chain of dependencies).
Batching updates removes the serial bottleneck of classical dynamic algorithms, and has therefore attracted considerable recent interest, with progress on batch-dynamic algorithms for several fundamental graph problems~\cite{GT24, acar2019batchconnect, yesantharao2021parallel,acar2020changeprop, dhulipala2019parallel, ghaffari2025, BB25a, ZKG+25, wang2020closest}.
In this paper, we extend this line of work by giving the first parallel batch-dynamic algorithm that maintains a proper $(\Delta+1)$-vertex coloring.

We state our main result below.

\begin{theorem}\label{thm:main-thm}
There exists a randomized algorithm that maintains a proper $(\Delta+1)$-vertex coloring in a dynamic $\Delta$-bounded graph using $O(\log \Delta)$ expected amortized work per update and, for any batch of $b$ updates, has parallel span $O(\poly b + \poly n)$ with high probability.
\end{theorem}

\subsection{Related Work} 

\subsubsection{Sequential setting.}
Above we mentioned three primary works in the sequential dynamic setting for $(\Delta+1)$-coloring and stated the best-known bounds achieved by Bhattacharya et al.~\cite{Bhattacharya-Coloring-Constant}.
We briefly review the results of the other two papers.

In~\cite{Dynamic-Coloring-18}, Bhattacharya et al.\ give the first dynamic algorithm for $(\Delta+1)$-coloring with $o(\Delta)$ update time.
In particular, they present a randomized algorithm that maintains a $(\Delta+1)$-vertex coloring in a $\Delta$-bounded graph with $O(\log \Delta)$ expected-amortized update time.
Since our main algorithm builds heavily on their approach, we provide a detailed overview of this result in Section~\ref{sec:BCHN-alg}.

In~\cite{Henzinger2019ConstantTimeD}, Henzinger and Peng give a dynamic algorithm for maintaining a proper $(\Delta+1)$-vertex coloring in a $\Delta$-bounded graph with $O(1)$ expected-amortized update time.
Whereas the $O(1)$ solution of Bhattacharya et al.~\cite{Bhattacharya-Coloring-Constant} builds on the $O(\log \Delta)$ solution of Bhattacharya et al.~\cite{Dynamic-Coloring-18}, Henzinger and Peng employ a quite different approach based on assigning random ranks to vertices.

\subsubsection{Parallel batch-dynamic setting.}
In the parallel batch-dynamic setting, existing dynamic coloring results are primarily arboricity-based.
In~\cite{LiuEtAl21BatchDynamicKCoreArxiv}, Liu, Shi, Yu, Dhulipala, and Shun introduce a parallel level data structure (PLDS).
Combined with the dynamic coloring framework of Henzinger et al.~\cite{Henzinger-arboricity-coloring}, they obtain a parallel batch-dynamic vertex-coloring algorithm that maintains an $O(\alpha \log n)$-coloring, where $\alpha$ is the arboricity of the graph, against an oblivious adversary using $O(\log^2 n)$ expected-amortized work per update and $O(\log^2 n \log\log n)$ span with high probability.

\begin{remark}
In this paper we focus on parallelizing the $O(\log \Delta)$ algorithm of Bhattacharya et al.~\cite{Dynamic-Coloring-18}.
At present, we do not know how to parallelize either of the $O(1)$-update-time algorithms~\cite{Bhattacharya-Coloring-Constant,Henzinger2019ConstantTimeD}.
Obtaining an $O(1)$ expected-amortized update time in the parallel batch-dynamic setting---either via a parallelization of one of these algorithms or via a new approach—is an interesting  problem which we leave open.
\end{remark}

%% file: sections/technical-overview.tex
\section{Technical Overview}
Here we outline the main technical ideas behind our parallel dynamic algorithm for $(\Delta + 1)$-coloring that has $O(\log \Delta)$ expected-amortized update time. The full details are given in Section \ref{sec:algorithm}. To build intuition for our algorithm, we start with a couple of warmups. Throughout this overview, we assume a $\Delta$-bounded input graph $G = (V,E)$ that is dynamically changing through a sequences of updates. Additionally, for ease of exposition, we forgo the details about the data structures maintained by these algorithms.

\subsection{Warmup I: a parallel dynamic algorithm for $2\Delta$-coloring.} To start, let us consider how we might maintain a $2\Delta$-coloring in the parallel batch-dynamic setting. We will see later that many of the ideas used here carry over to the $(\Delta + 1)$-coloring algorithms discussed in Warmup II.
 
\subsubsection{A sequential algorithm with $O(1)$ expected update time.}
It turns out that there is a rather simple folklore algorithm to maintain a $2\Delta$ coloring in the sequential setting using only $O(1)$ expected time. Fix a color palette $\mc C$ of $2\Delta$ colors and let $\chi$ denote the dynamic coloring we aim to maintain. Each vertex $u$ stores the last time $\tau_u$ at which it was recolored. Now suppose that $\chi$ is proper at time $\tau$ and consider an edge update $e = (u,v)$. If $e$ is a deletion or $\chi(u) \not= \chi(v)$, then we do nothing. Otherwise, when $\chi(u) = \chi(v)$, we recolor the endpoint of $e$ that was most recently recolored. Without loss of generality, assume this endpoint is $u$, so $\tau_u > \tau_v$. To recolor $u$, we first compute the set of \emph{blank} colors for $u$, that is the colors in $\mc C$ which are not used by any neighbor of $u$. Denote this set by $B_u$. The new color for $u$ is then choose uniformly at random from $B_u$. By construction, none of $u$'s neighbors (including $v$) have the same color as $u$ and thus $\chi$ is now proper.

To analyze the expected time required for update $e$, we use the principle of deferred decision. The only non-trivial work is done in that case that $e$ is monochromatic, in which case computing $B_u$ requires $O(\Delta)$ work. Observe that $e$ is monochromatic at time $\tau$ only if $u$ was colored with $\chi(v)$ at time $\tau_u$. To compute the probability of this event we condition on all the random choices made by the algorithm up until just before time $\tau_u$. Since $\tau_u > \tau_v$ this fixes $\chi(v)$. At time $\tau_u$, $u$ chooses its color uniformly at random from $B_u$ and thus the probability that $\chi(u) = \chi(v)$ is at most $1/|B_u| \leq 1/\Delta$ (since $|\mc C| = 2\Delta$ and $\deg(u) \leq \Delta$). Therefore, the expected time spent on the addition of edge $e$ is $O(\Delta) \cdot 1/\Delta = O(1)$. Note that this analysis crucially relies on the fact that the set of edge updates is oblivious to the choices of the algorithm, that is the adversary is assumed to be oblivious. 

\subsubsection{Parallelizing the sequential algorithm.} Now let us consider how to parallelize this algorithm. As before, assume that at time $\tau$ our coloring $\chi$ is proper. In the parallel setting, our update is now a batch of edge updates $S$. First compute the set of inserted edges that are monochromatic, i.e. compute $S_\text{mon} = \{(u,v) = e \in S : \chi(u) = \chi(v)\}$. For each of these edges we will need to recolor at least one of their endpoints. As in the sequential version, we will recolor the endpoints that were most recently recolored. Denote this set by  $V_\blank := \{u : e = (u,v) \in S_{\text{mon}} \text{ and } \tau_u \geq \tau_v\}$. As a first step, we remove the color from each vertex in $V_\blank$. We refer to these vertices as \emph{blank} vertices and denote the \emph{blank subgraph} induced by them as $G[V_\blank]$. We are now left with the task of coloring $G[V_\blank]$.

To color the blank subgraph $G[V_\blank]$ we will use a static parallel $(\deg+1)$-coloring algorithm which is described next.

\paragraph{A parallel $(\deg+1)$-coloring subroutine.}
We give a parallel $(\deg+1)$-coloring algorithm to color a graph $H$. Each vertex $u$ has a palette $C_u$ of available colors with $|C_u| \ge \deg_H(u) + 1$. The algorithm proceeds in a sequence of rounds in which every vertex independently samples a tentative color from its current palette and keeps it only if no neighbor chose the same color. At the end of each round the colored vertices are removed and the uncolored vertices update their palettes. In Section~\ref{sec:partial-coloring-alg}, we show that this procedure colors a constant fraction of the remaining vertices in each round, runs in $O(|E(H)|)$ expected work, and has $O(\log^2 n)$ span with high probability.

\paragraph{Coloring the blank subgraph.} We color $G[V_\blank]$ using the parallel $(\deg + 1)$-coloring routine described above. We instantiate each vertex's color palette with $B_u$, so $C_u := B_u$. Let $c_u$ be the color assigned to the blank vertex $u$ during the $(\deg + 1)$-coloring routine. We make two observations:
\begin{enumerate}
    \item At the start of the routine, $|C_u| \geq \deg_{G[V_\blank]}(u) + \Delta$.
    \item Throughout the routine, $|C_u| \geq \Delta$.
\end{enumerate}
The first observation implies the second and the second observation implies that $c_u$ was chosen randomly from a palette of blank colors of size at least $\Delta$. We prove $(1)$ in the following claim.

\begin{claim}\label{clm:many-blank-colors-2delta}
For every $u \in V_\blank$ we have
\[
  |C_u| \ge \deg_{G[V_{\blank}]}(u) + \Delta.
\]
\end{claim}
\begin{proof}
By definition $C_u = B_u$ so it suffices to show $|B_u| \geq \deg_{G[V_{\blank}]}(u) + \Delta.$ Fix a vertex $u \in V_\blank$. Let $t$ be the number of distinct colors that appear among the colored neighbors of $u$. A color is not in
$B_u$ exactly if it appears on at least one neighbor of $u$, so
\[
  |B_u| = |\mc C| - t = 2\Delta - t.
\]
The number of colored neighbors of $u$ upper bounds $t$, and $\deg(u) \le \Delta$, so
\[
  t \le \deg(u) - \deg_{G[V_{\blank}]}(u) \le \Delta - \deg_{G[V_{\blank}]}(u).
\]
We then obtain
\[
  |B_u|
    = 2\Delta - t
    \ge 2\Delta - (\Delta - \deg_{G[V_{\blank}]}(u))
    = \Delta + \deg_{G[V_{\blank}]}(u),
\]
as claimed.
\end{proof}

Finally, we show that this parallel algorithm uses $O(|S|)$ expected work for a batch $S$, which translates to $O(1)$ expected work per update. Identifying the set $S_{\text{mon}}$ of monochromatic inserted edges takes $O(|S|)$ work. Additionally by our second observation above, each vertex samples its new color from a palette of size at least $\Delta$. Therefore, for each $e = (u,v) \in S$, the same deferred decision argument as in the sequential case implies that
\[
  \Pr[e \in S_{\text{mon}}] \;\le\; 1/\Delta,
\]
so $\E[|S_{\text{mon}}|] \le |S|/\Delta$. Each monochromatic edge contributes at most one vertex to $V_\blank$, so $|V_\blank| \le |S_{\text{mon}}|$ and hence $\E[|V_\blank|] = O(|S|/\Delta)$. For each $u \in V_\blank$, computing $B_u$ and participating in the static $(\deg+1)$-coloring routine requires $O(\deg(u)) = O(\Delta)$ work, and the total work of the subroutine on $G[V_\blank]$ is $O(|E(G[V_\blank])|) = O(\Delta |V_\blank|)$. Taking expectations, this is
\[
  \E[\,O(\Delta |V_\blank|)\,] \;=\; O\bigl(\Delta \cdot \E[|V_\blank|]\bigr) \;=\; O(|S|),
\]
so the parallel warmup algorithm maintains a $2\Delta$-coloring using $O(1)$ expected work per edge update under an oblivious adversary.

\subsection{Warmup II: a sequential dynamic algorithm for $(\Delta + 1)$-coloring.} We next turn to the problem of $(\Delta + 1)$-coloring. To start, we will recall the the sequential $(\Delta + 1)$-coloring algorithm of Bhattacharya et al.~\cite{Dynamic-Coloring-18} which has an expected-amortized update time of $O(\log \Delta)$. Going forward, we refer to this as the BCHN algorithm. Then we propose a new relaxed variant of the BCHN algorithm which we will ultimately parallelize in section \ref{sec:tech-sec-3}.

\subsubsection{The BCHN algorithm.}\label{sec:BCHN-alg} The main obstacle in moving from $2\Delta$ to $\Delta+1$ colors is that a vertex $u$ may have only one blank color. Therefore, if we only use blank colors for recoloring, the same deferred decision argument used in the previous algorithms no longer applies. On the other hand, recoloring $u$ uniformly from all $\Delta+1$ colors can create many conflicts. BCHN takes a middle ground: when recoloring $u$, it samples uniformly from the set of colors that are either blank for $u$ or unique among (a carefully chosen subset of) the neighbors of $u$. In the case that a unique color is sampled, the corresponding neighbor is recursively recolored by the same process. The neighbors are chosen so that (1) the sampling space is large (so future conflicts are unlikely against an oblivious adversary), (2) any recoloring step creates at most one new conflict, so the recolorings form a chain rather than a branching cascade, and (3) the cost of coloring each vertex in the chain decreases geometrically. To do this, BCHN make use of a structure called a \emph{hierarchical partition}. The cost of maintaining this structure will be $O(\log \Delta)$ amortized, but will allow efficient recoloring according to the procedure above.

\paragraph{Hierarchical partition (HP).}
Fix a sufficiently large constant $\beta>1$ and let $\lambda=\Theta(\log_\beta \Delta)=\Theta(\log\Delta)$.
Each vertex $u$ is assigned a level $\ell(u)\in\{1,\dots,\lambda\}$, and for indices $i\le j$ we write
\[
  N_u(i,j) := \{v \in N(u) : i \le \ell(v) \le j\}.
\]
We call $N_u(\ell(u),\lambda)$ the \emph{up-neighbors} of $u$ and $N_u(1,\ell(u)-1)$ its \emph{down-neighbors}. The HP maintains the following two conditions:
\begin{condition}[Upper]\label{con:upper:bchn}
  $\forall u\in V:\quad |N_u(1,\ell(u))| \le \beta^{\ell(u)}.$
\end{condition}
\begin{condition}[Lower]\label{con:lower:bchn}
  $\forall u\in V \text{ with }\ell(u)>1:\quad |N_u(1,\ell(u)-1)| \ge \beta^{\ell(u)-5}$.
\end{condition}
A vertex is \emph{clean} if it satisfies both conditions and \emph{dirty} otherwise. After each edge update, BCHN runs a greedy \textsc{Maintain-HP} procedure until all vertices are clean: a vertex violating the Upper condition is moved up to the lowest level (above the current level) where Upper holds; otherwise, a vertex violating Lower is moved down to the highest level (below the current level) where it satisfies a strong lower bound $|N_u(1,k-1)|\ge \beta^{k-1}$ (or to the bottom level if no such $k$ exists). This creates \emph{slack} which is crucial for the amortized analysis: if $u$ is moved to level $k>1$, then immediately after the move it has at least $\beta^{k-1}$ down-neighbors, but it can move down only after dropping below $\beta^{k-5}$, so it must lose $\Omega(\beta^{k-1})$ down-neighbors before moving down again.

Using the appropriate data structures, a vertex at level $i$ can be moved to level $k$ using $O(\beta^{\max\{i,k\}})$ work. To pay for this work, a careful accounting using token functions is required. Each vertex and edge has an associated token function, where upon any edge update, the total number of tokens in the system increases by at most $\lambda$. Additionally, by how the movements are chosen, each movement of a vertex from level $i$ to level $k$ releases $\Omega(\beta^{\max\{i,k\}})$ tokens from the system. This leads to the $O(\lambda) = O(\log \Delta)$ amortized update time.

\paragraph{Color palettes.}
Let $\mc C=\{0,1,\dots,\Delta\}$ be the global palette.
For a vertex $u$ at level $i=\ell(u)$, define
\[
  \mc C_u^{+} := \{\chi(v): v\in N_u(i,\lambda)\}
  \qquad\text{and}\qquad
  \mc C_u^- := \mc C\setminus \mc C_u^{+}.
\]
Thus $\mc C_u^-$ is the set of colors \emph{not used by up-neighbors}. A color $c\in \mc C_u^-$ is:
(i) \emph{blank} for $u$ if no down-neighbor has color $c$; and
(ii) \emph{unique} for $u$ if exactly one down-neighbor has color $c$.
Let $B_u$ and $U_u$ denote the blank and unique colors, respectively. A key fact is that there is always many blank plus unique colors:

\begin{claim}[Blank+unique colors]\label{clm:bchn-bu}
For any vertex $u$ at level $i$,
\[
  |B_u\cup U_u| \;\ge\; 1 + \frac{|N_u(1,i-1)|}{2}.
\]
\end{claim}
\begin{proof}
    Define $T_u := \mc C_u^- \setminus (B_u \cup U_u)$. The claim follows from the following observations: $(1)$ $|\mc C_u^-| \geq 1 + |N_u(1, i-1)|$, $(2)$ $|\mc C_u^-| = |B_u \cup U_u| + |T_u|$, and $(3)$ $2|T_u| \leq |N_u(1, i - 1)|$.
\end{proof}

In particular, if $u$ satisfies the Lower condition at level $i$, then
$|B_u\cup U_u| = \Omega(\beta^{i-5})$.

\paragraph{Recoloring routine.} As before, the
BCHN algorithm maintains a timestamp $\tau_u$ recording the last update at which $u$ was recolored. When a vertex $u$ must be recolored, BCHN first computes $B_u\cup U_u$, and samples a new color $c$ uniformly at random from this set. If $c\in B_u$ we stop. If $c\in U_u$, then there is a unique down-neighbor $v$ with $\chi(v)=c$; the algorithm recursively recolors $v$ in the same way. Because recursive calls always go to strictly lower levels, the recursion depth is at most $\lambda$, and the total time for recoloring $u$ (including the whole chain) is
\[
  O\!\left(\sum_{h\le \ell(u)} \beta^{h}\right) = O(\beta^{\ell(u)}),
\]
where the per-level work comes from scanning/maintaining neighborhood information which is in turn bounded by the Upper condition. 

\paragraph{Full update algorithm.}
On an edge deletion, BCHN only updates the HP.
On an insertion of $(u,v)$, it first runs \textsc{Maintain-HP}. If after that $\chi(u)\neq \chi(v)$,
it does nothing further. Otherwise it recolors the endpoint with larger timestamp (the one recolored
more recently), say $u$, by calling the recoloring routine described above. 

\paragraph{BCHN Analysis.} It is tempting (but incorrect) to write $\Pr[\text{$(u,v)$ is monochromatic}] \le 1/|B_u\cup U_u| = 1/\Omega(\beta^{\ell(u)})$
where $\ell(u)$ is the \emph{current} level when the edge is inserted. The issue is that the random choice that produced $\chi(u)$ was made at time $\tau_u$, potentially when $u$ was at a \emph{different} level.

The correct analysis tracks both levels:
let $i$ be the level of $u$ at the time of the conflicting insertion (after \textsc{Maintain-HP}), and let
$j$ be the level of $u$ at time $\tau_u$ when it last sampled its current color.
BCHN splits the analysis into cases:
\begin{itemize}
  \item If $i\le j$ (the level did not increase since the last recoloring), then by deferred decisions
        the insertion creates a conflict with probability at most $1/|B_u\cup U_u|$ evaluated at time $\tau_u$,
        which is $1/\Omega(\beta^{j-5})$ by Claim~\ref{clm:bchn-bu} and the Lower condition.
        Since recoloring $u$ costs $O(\beta^{i})\le O(\beta^{j})$, the expected cost is
        $O(\beta^{i})\cdot O(1/\beta^{j-5}) = O(1)$.
  \item If $i>j$ (the level increased), then the expensive recoloring at level $i$ is not paid for by
        the above probability bound. Instead, BCHN charges it to the work already spent by \textsc{Maintain-HP}
        in moving $u$ up to level $i$ (which is $\Omega(\beta^{i})$), and that cost is covered by the HP token
        accounting.
\end{itemize}
Together with the $O(\log\Delta)$ amortized cost of maintaining the HP, this yields $O(\log\Delta)$
expected-amortized update time for maintaining a $(\Delta+1)$-coloring against an oblivious adversary.

\paragraph{Obstacles to parallelization.}
The main challenge in parallelizing BCHN is the \textsc{Maintain-HP} routine. Moving a single vertex changes the up-neighbor and down-neighbor sets of many other vertices, which can cause them to violate the Upper and Lower conditions and trigger further moves. This creates long chains of dependent level changes whose order matters, making it unclear how to process many updates (or even a single update) with significant parallelism.

\subsubsection{A new sequential dynamic algorithm for $(\Delta +1)$-coloring.}\label{sec:tech-sec-2} To avoid the obstacles to parallelization posed by the BCHN algorithm, we take a slightly different viewpoint on the role of the hierarchical partition. First, we summarize the role of the HP in the BCHN algorithm:
\begin{enumerate}
  \item For \emph{clean} vertices, the Upper and Lower conditions imply that $|B_u\cup U_u|$ is large, so the deferred-decision argument can bound the probability that a future edge insertion creates a conflict with $u$.
  \item The Upper condition bounds $|N_u(1,\ell(u))|$ and thus controls the work needed to recolor $u$ at level $\ell(u)$: maintaining and scanning the relevant neighborhood information costs $O(\beta^{\ell(u)})$ work.
  \item Because recolorings follow unique colors down the levels, a recoloring chain from a vertex at level $i$ visits strictly lower levels, and the work per vertex decreases geometrically with the level. As a result, the total cost of a chain is dominated (up to constants) by the cost of recoloring its first vertex.
\end{enumerate}
In the BCHN algorithm these effects are enforced via the \textsc{Maintain-HP} routine, which ensures every vertex is clean after each update. However, as discussed above, this procedure is the main barrier to parallelization.

A useful observation from the BCHN analysis is that there are multiple ways to pay for recoloring a vertex $u$:
\begin{enumerate}
  \item If the level of $u$ at the time of the conflicting insertion is at most its level when it last sampled its current color, then the deferred-decision argument applies and the expected recoloring cost is $O(1)$, using the size of $B_u\cup U_u$ at the earlier time.
  \item If the level of $u$ has increased since its last recoloring, then this argument no longer applies. However, BCHN charge the recoloring cost to the token potential that was already spent in moving $u$ up; the same token drop that pays for the movement suffices to pay for recoloring at the new level.
\end{enumerate}

This suggests that globally maintaining a clean hierarchical partition may be unnecessary. We can afford to let vertices become dirty and only ``repair'' them when we actually encounter them during recoloring. Moreover, when we do repair a dirty vertex, we do not need to move it all the way to a fully clean level; it suffices to move it \emph{enough} so that the resulting decrease in the token potential pays for (1) the movement itself and (2) the cost of recoloring that dirty vertex. Of course, allowing dirty vertices to be recolored complicates the deferred decision argument used previously, but we will see that it can be made to work.

In the remainder of this section we build on this intuition to a design a new relaxed algorithm and show that it achieves the same $O(\log \Delta)$ expected-amortized update time as the BCHN algorithm.

\paragraph{Algorithm.}
The pseudocode is given in Algorithm~\ref{alg:relaxed-seq-1}.
Consider again the insertion of a monochromatic edge $e=(u,v)$ at time-step $\tau$, and let $x$ be the endpoint that was most recently recolored. Upon the insertion we invoke $\textsc{Recolor}(x)$. In our algorithm, we retain BCHN's randomized recoloring step \emph{only} for clean vertices. If the current vertex $y$ is clean, we sample uniformly from $B_y \cup U_y$ and recurse only if we pick a unique color, so the recursion continues as a single chain down the levels.

If the chain ever reaches a dirty vertex $y$, we recolor $y$ deterministically with an arbitrary blank color (which always exists with $\Delta+1$ colors), and then invoke the \textsc{Move} procedure to shift $y$ enough to release a sufficient number of tokens to pay for both the movement and the deterministic coloring. We also mark the timestamp of $y$ as deterministic so that later we can distinguish these recolorings from random ones in the analysis. By construction, the recursion depth of this procedure is at most $\lambda$.

\begin{algorithm}[ht!]
\function{\emph{\textsc{Update}$(u,v)$}}{
    \If{$(u,v)$ is inserted and $\chi(u) = \chi(v)$}{
        $x \gets \arg\max_{y \in \{u,v\}} \tau_y$. \\
        $\textsc{Recolor}(x)$.
    }
}
\function{\emph{\textsc{Recolor}$(u)$}}{
    \If{$u$ is clean}{
        Choose $c$ uniformly at random from $B_u \cup U_u$.\\
        $\chi(u) \gets c$. \\
        \If{$c$ is unique}{
            Find the unique vertex $v \in N_u(1, \ell(u) - 1)$ with $\chi(v) = c$.\\
            $\textsc{Recolor}(v)$.
        }
    } \Else{
        Choose $c$ from $B_u$.\\
        $\chi(u) \gets c$. \\
        $\textsc{Move}(u)$.
    }
}

\function{\emph{\textsc{Move}$(u)$}}{
    \If{$u$ does not satisfy the upper condition \ref{con:upper:bchn}}{
        $k \gets \lambda$.\\
        \While{$|N_u(1,k-1)| < \beta^{k-1}$ or $|N_u(1,k)| > \beta^{k}$}{$k \gets k - 1$.}
        Move $u$ up to level k. \\
    }
    \Else{
        Move $u$ down to level $\ell(u) - 4$.
    }
}
\caption{Relaxed Sequential Algorithm}
\label{alg:relaxed-seq-1}
\end{algorithm}

\paragraph{Analysis.} Now we analyze the update time of our sequential algorithm. Consider a sequence of $T$ edge insertions/deletions starting from an empty graph. Let $W_T$ denote the the total work done over the sequence of updates. We want to show that $\E[W_T] = O(T \log \Delta)$.

The strategy at a high level is as follows.
\begin{enumerate}[(1)]
    \item First, we split the total work into the individual work done by each recoloring recursion chain. The cost of a chain is dominated by two parts: that of a spawning clean vertex and of a terminating dirty vertex.
    \item Then, we split the spawning clean vertices into two types: ones which were clean when recolored last and ones which were dirty when recolored last.
    \item To understand the cost of the former type, we use a variant of the deferred decision argument, giving an $O(1)$ cost per update. We relate the cost of the latter type to the cost of all dirty vertices.
    \item We show the cost of recoloring dirty vertices is dominated by the movement cost.
    \item Finally, we show that moving dirty vertices releases an amount of tokens proportional to the cost of the movement, giving an $O(\log \Delta)$ amortized cost per update.
\end{enumerate}

To start, we introduce the notion of a \emph{record}. During the execution of the algorithm, every time a vertex $u$ is recolored, we record the timestamp in a sequence. We call this sequence the \emph{record} of $u$, and denote it by $r_u = (\tau_u^0, \tau_u^1, \dots, \tau_u^k)$. Each timestamp $\tau$, depending on when it was recorded, has a cost $w(\tau)$, corresponding to the recoloring and potential moving cost of the associated vertex. If at $\tau_u^i$, the vertex $u$ is clean, we have
\[
    w(\tau_u^i) = O(\beta^{\ell(u)}),
\]
and if $u$ is dirty, we have
\[
    w(\tau_u^i) = \Theta(\beta^{\max\{\ell(u),k\}}),
\]
where $k$ is the level $u$ ends up on.

We can interpret the work of the algorithm as
\[
    W_T = \sum_u \sum_{\tau \in r_u(T)} w(\tau).
\]
Only some of the timestamps in the records contribute meaningful work, which we isolate now. The execution of the recoloring algorithm naturally corresponds to a chain of recolorings of vertices at monotonically decreasing levels. Such a chain has a spawning vertex and a terminating vertex. The total cost of the records of clean vertices in a chain is proportional to cost of the clean spawning vertex,  since in a record the levels of vertices in a chain are strictly decreasing and thus the work is a geometric sum. the total cost of the records of clean. We can conclude that the overall cost of the timestamps in a chain is dominated by the cost of the timestamp of the spawning vertex (if it is clean) together with the cost of the timestamp of the terminating vertex (if it is dirty).

We thus see that the cost of the entire record is proportional to the cost of the \emph{important} timestamps $I$, where an \emph{important} timestamp is one which corresponds to a clean spawning vertex or a dirty terminating vertex. According to this distinction, we partition $I$ into clean and dirty sets denoted $C_T$ and $D_T$. We may further partition $C_T$ into two sets $CC_T$ and $DC_T$, where the first corresponds to timestamps $\tau_u^i \in C_T$ such that $\tau_u^{i-1} \in C_T$, and the latter corresponds to $\tau_u^i \in C_T$ such that $\tau_u^{i-1} \in D_T$. One final observation is that $|I|$ is at most $2T$, since each chain gives at most $2$ important timestamps. 

We now state a lemma that is useful in understanding the cost of the record.
\begin{lemma}
    If $\tau_u^i \in C_T$ such that $\tau_u^{i-1} \in D_T$, then $w(\tau_u^i) \leq w(\tau_u^{i-1})$.
\end{lemma}
\begin{proof}
    Since $\tau_u^{i-1}$ is dirty, the cost of $\tau_u^{i-1}$ is $\Omega(\beta^{\max\{\ell_{i-1}(u),\ell_{i}(u)\}})$, since $\ell_{i}(u)$ is the level $u$ ends up on after moving, and the cost of the clean timestamp $\tau_u^{i-1}$ is $O(\beta^{\ell_i(u)})$.
\end{proof}

As a corollary, we immediately get that $w(DC_T) \leq w(D_T)$. Thus we can bound the total work as
\[
    W_T \leq w(CC_T) + 2w(D_T).
\]
The first term will be understood using a \textbf{deferred decision argument}. We will understand the second term by relating it to the work done solely in \textbf{moving dirty vertices}, which will be analyzed using a token function argument.

\paragraph{The clean part.}
We understand this part using a variant of the deferred decision argument described in the BCHN algorithm. To do this, we interpret the cost of the clean part from the perspective of edge updates. We define the timestamp $\tau_e$ to be the timestamp that would result if the edge $e$ would be monochromatic upon insertion. We also let $\tau_e^-$ denote the predecessor of $\tau_e$ in its appropriate record. We then have the equality
\begin{equation}
    \label{eq:work-cc}
    w(CC_T) = \sum_{e \in S_T} \ind[\tau_e \in CC_T] \cdot w(\tau_e).
\end{equation}
We use this equation to prove the following equality (Lemma \ref{lem:def-dec} in the main part of the paper):
\[\E[w(CC_T)] = O(T).\]

\paragraph{The dirty part.}
We now recall the following token functions used in the BCHN algorithm. For edges $(u,v)$ and vertices $u$, we define
\begin{equation}\label{eqn:BCHN-vertex-potential}
        \theta(u, v) = \lambda - \max(\ell(u), \ell(v)).
\end{equation}   
\begin{equation}\label{eqn:BCHN-edge-potential}
    \theta(u) =
    \begin{cases} \frac{\max(0, \beta^{\ell(u) - 1} - |N_u(1, \ell(u) - 1)|)}{2\beta} & \text{if } \ell(u) > 1; \\
            0 & \text{otherwise}.
    \end{cases}
\end{equation}
We use $\Gamma$ to denote the total number of tokens in the system. It is easy to see that an edge update can increase the total number of tokens by at most $\lambda$. The critical use of the token functions is to associate the work done in moving vertices to a proportional decrease in $\Gamma$. We note that coloring a dirty vertex from $B_u$ and moving it from level $i$ to level $k$ can be done in $O(\beta^{\max(i,k)})$ work using appropriate data structures. We will show that moving a vertex from level $i$ to level $k$ releases $\Omega(\beta^{\max(i,k)})$ tokens. Overall, this will imply that \[
w(D_T) = O(T \log \Delta).\]
\begin{proof} 
There are two cases we must consider.

\begin{description}\label{desc:token-arg}
  \item[\textbf{Case 1: upward movement.}] 
  In this case, suppose we have a vertex $u$ at level $i$, which
  violates the upper condition and is moved (and colored) to level $k$ during Algorithm \ref{alg:relaxed-seq-1}. The cost of this is $O(\beta^{k})$. We will show that $\Gamma$ decreases by $\Omega(\beta^{k})$. Denote by $\theta'$ the token values after moving. By construction, prior to moving, $u$ satisfies \begin{enumerate*}
    \item $|N_u(1,k-1)| \geq \beta^{k-1}$, and
    \item $|N_u(1,5)| \leq \beta^{k}$.
\end{enumerate*} Then by $(1)$ $\theta'(u) = 0$. Additionally, for each $x \in N_u(i+1,k)$, $\theta'(x) \leq \theta(x) + 1/2\beta$. So the total number of tokens associated with neighbors of $u$ increased by at most $|N_u(i+1,k)|/2\beta \leq |N_u(1, k)|/2\beta \leq \beta^{k-1}/2$. Next, for each $x \in N_u(j)$ where $1 \leq j < k$, $\theta'(x, u) \leq \theta(x,u) - (k-j)$. Thus we have the total decrease in tokens associated to edges is at least $|N_u(1,k-1)| \geq \beta^{k-1}$ by $(1)$. Thus in total, $\Gamma$ decreases by at least $\beta^{k-1} - \beta^{k-1}/2 = \beta^{k-1}/2 = \Omega(\beta^k)$.

  \item[\textbf{Case 2: downward movement.}]
  In this case, suppose we have a vertex $u$ at level $i$, which satisfies the upper condition, but violates the lower condition. Lowering a vertex does not increase any vertex token function by definition \ref{eqn:BCHN-vertex-potential}. Moreover, initially $\theta(u)$ was $(\beta^{i-1} - |N_u(1,i-1)|)/2\beta \geq (\beta^{i-1} - \beta^{i-5})/2\beta$ as $u$ violated the lower condition at level $i$. After, the movement, $u$ is at level $i-4$, and thus $\theta'(u) \leq \beta^{i-6}/2$. Therefore, the vertex token decrease is at least $\beta^{i-6}(\beta^4 - 1)/2$. Now, we analyze the token change from edges. From \ref{eqn:BCHN-edge-potential}, the number of tokens associated with the edge $(u,v)$ for each $v \in N_u(1,i-1)$, increases by $i - \max(\ell(v), i-4) \leq 4$. Therefore, we can upper bound the total increase in edge tokens by $4|N_u(1 \ell(u)-1)| < 4 \cdot \beta^{i-5}$. Thus in total, $\Gamma$ decreases by at least $\beta^{i-6}(\beta^4 - 1)/2 - 4\cdot \beta^{i-5} = \Omega(\beta^i)$ for sufficiently large $\beta$.

\end{description}
This establishes the inequality.
\end{proof}

\paragraph{Finalizing the argument.}
Putting together the clean and dirty parts, we conclude that
\[
    \E[W_T] = O(T \log \Delta).
\]

\subsection{A parallel dynamic algorithm for $(\Delta+1)$-coloring.}\label{sec:tech-sec-3}
To obtain a parallel dynamic algorithm our aim is to parallelize the two primary subroutines in the relaxed sequential algorithm: the recoloring routine, and the movement routine. Assuming we have the parallelizations of these subroutines, the goal is to run the same argument as in the relaxed sequential algorithm. We observe that, at a high level, the argument only makes use of the following two key properties of the movement and recoloring routines:
\begin{enumerate}[(1)]
    \item The cost of recoloring dirty vertices is dominated by the cost of moving those same vertices.
    \item The movement of dirty vertices releases an amount of tokens proportional to the cost of the movement.
\end{enumerate}
Thus, if we can design parallel versions of each subroutine so that the two above facts hold, we will immediately obtain an $O(\log \Delta)$ expected-amortized algorithm.

\begin{remark} As in the previous parts of the technical overview, we shall leave $\beta$ for the time being. In the full algorithm (section \ref{sec:algorithm}) we will use $\beta = 3$.
\end{remark}

\subsubsection{Algorithm overview.} Our algorithm parallelizes the relaxed sequential algorithm given in the previous section in three phases: the initialization phase, the coloring phase, and the moving phase. We outline each of these phases and provide the main conceptual ideas behind the algorithm.

First however, we remark that there is an important difference in the conditions we use for the hierarchal partition. Specifically, we will use the following relaxed upper condition which will be crucial to the final analysis.
\begin{condition}[Upper]\label{con:upper:new}
    For every vertex $u \in V$, we have $|N_u(1,\ell(u))| \leq \beta^{\ell(u)+2}$.
\end{condition}
The lower condition remains the same. 

\subsubsection{Initialization phase.} In the parallel setting, our updates now come in the form of a batch $S$. Suppose that $S$ is a batch of edge insertions (since deletions do not trigger recolors). As in the relaxed sequential algorithm, our first step is to identify which vertices need to be uncolored after inserting $S$. We do this in the same way as in the parallel $2\Delta$-coloring algorithm from warmup II. Concretely, we compute the subset of edges in $S$ that are monochromatic and within that set compute the endpoints that were most recently recolored, that is we compute $V_\text{blank}$.  We would also like to distinguish between the dirty and clean blank vertices and further separate dirty blank vertices into those that violate the upper condition and those that only violate the lower condition. Formally, we denote
\begin{itemize}
    \item $\upmrkd(i)$, set of blank vertices on level $i$ which don't satisfy the upper condition \ref{con:upper:new}.
    \item $\lmrkd(i)$, set of blank vertices on level $i$ which satisfy the upper condition \ref{con:upper:new} but don't satisfy the lower condition \ref{con:lower:bchn}. 
    \item $\umrkd(i)$, set of blank vertices on level $i$ which are clean.
\end{itemize}
We collectively refer to the vertices in $\upmrkd(i) \cup \lmrkd(i)$ as \emph{marked} and the vertices in $\umrkd(i)$ as \emph{unmarked}. Collectively, these sets can be computed in $O(|S|)$ expected work and $O(\log n)$ span \whp. 

\subsubsection{Coloring phase.} After the initialization phase, all conflicts created by a batch $S$ of updates are in a set $V_{\blank}$ of blank vertices, whose colors have been temporarily removed. We now need to assign new colors to $V_{\blank}$ while preserving a proper $(\Delta+1)$-coloring. We process $V_{\blank}$ level-by-level from the top of the partition downwards. Additionally, we write
\[
  \mrkd(i) := \upmrkd(i)\cup \lmrkd(i)
\]
for the set of \emph{marked} vertices on level $i$. 

The coloring phase on level $i$ proceeds in two steps:
first we color $\mrkd(i)$, then we color $\umrkd(i)$. In both cases we reduce the task to repeated applications
of the parallel partial list-coloring routine from Section~\ref{sec:partial-coloring-alg}.

At a high level, a coloring procedure on some set $X\subseteq V_{\blank}$ on level $i$ works as follows.
We initialize a coloring instance $(H_0,C_{H_0})$ with $H_0=G[X]$ and an admissible palette $C_{H_0}(u)$ for each
$u\in X$. Given an instance $(H_t,C_{H_t})$, we run the partial coloring routine, which returns a set $L$ of newly
colored vertices and a residual graph $H_{t+1}$ on the remaining uncolored vertices. Using $L$, we compute a new palette  $C_{H_{t+1}}(u)$ for each $u\in V[H_{t+1}]$, and iterate until $H_t$ becomes empty. The analysis of
Section~\ref{sec:partial-coloring-analysis} shows that as long as every instance satisfies
\[
  |C_{H_t}(u)| \;\ge\; \deg_{H_t}(u) + 1\quad\text{for all $u$,}
\]
each round colors a constant fraction of the remaining vertices in expectation, so we obtain $O(\log^2 n)$ span
per level \whp.

The main difference between coloring $\mrkd(i)$ and $\umrkd(i)$ is in how we choose the color palettes and the coloring guarantees we need to ensure. 

\paragraph{Marked vertices.} In the relaxed sequential algorithm, a recoloring chain stops as soon as it reaches a marked (dirty) vertex, which is then recolored deterministically with a blank color and subsequently moved. In the parallel setting we mimic this behavior by coloring all vertices in $\mrkd(i)$ using only blank colors.

Concretely, for level $i$ we set $H_0 := G[\mrkd(i)]$. For each $u\in \mrkd(i)$ we preprocess its dynamic palette $\mathcal C_u$ so that the lower palette $\mathcal C_u^-$ begins with exactly the colors that are
blank for $u$. A simple counting argument, analogous to the warmup $2\Delta$-coloring case, shows that $u$ always has at least $\deg_{H_0}(u)+1$ blank colors available. We therefore define
\[
  C_{H_0}(u) := \mathcal C_u^-[0:\deg_{H_0}(u)].
\]
and run the partial coloring routine on $(H_0,C_{H_0})$. The vertices in $L$ receive blank colors, so they do not create new conflicts and do not trigger further recoloring. We then update the data structures and shrink each $C_{H_0}(u)$ by removing colors used by neighbors in $L$, obtaining $(H_1,C_{H_1})$. Iterating this process until $H_t$ is empty colors all marked vertices on level $i$ using only blank colors and maintains the invariant $\lvert C_{H_t}(u)\rvert \ge \deg_{H_t}(u)+1$ throughout.

\paragraph{Unmarked vertices.} Here we want to preserve the deferred-decision argument: whenever an unmarked vertex $u$ is recolored, it should choose its new color uniformly from a palette of blank-or-unique colors $B_u\cup U_u$ that is large both compared to its lower neighborhood and compared to its current degree in the induced subgraph of still-uncolored vertices.

Fix a level $i$ and let $H_0 := G[\umrkd(i)]$. For each $u\in \umrkd(i)$ we again preprocess $\mathcal C_u^-$, but now we push to the end all colors that are neither blank nor unique on $N_u(1,i-1)$. The remaining prefix of $\mathcal C_u^-$ consists exactly of colors in $B_u\cup U_u$. By the palette lower bound for clean vertices (Claim~\ref{clm:bchn-bu} and its analogue in our setting), this prefix has size
\[
  |B_u\cup U_u| \;\ge\; \Theta\bigl(|N_u(1,i-1)|\bigr),
\]
and in fact we ensure that
\[
  |C_{H_0}(u)| \;\ge\; \max\!\bigl\{\deg_{H_0}(u)+1,\; |N_u(1,i-1)|/2\bigr\}.
\]
We then set $C_{H_0}(u)$ to be the first
$\max(\deg_{H_0}(u)+1, |N_u(1,i-1)|/2+1)$ entries of this prefix and run the partial coloring routine on
$(H_0,C_{H_0})$.

In the resulting rounds, each $u\in \umrkd(i)$ that is colored picks a color uniformly from a large set of blank-or-unique colors. If a blank color is chosen, recoloring stops at $u$. If a unique color $c$ is chosen, there is a single lower neighbor $v$ with $\chi(v)=c$; we make $v$ blank and place it into the appropriate marked/unmarked set on its level 

A technical point is that, unlike in the marked case, the palettes $C_{H_t}(u)$ are not easy to maintain incrementally once some neighbors in $H_t$ are colored: removing a neighbor may change which colors are unique. Instead, we simply recompute $C_{H_t}(u)$ from scratch at the beginning of each round on the current residual subgraph $H_t$. Because all vertices in $\umrkd(i)$ are clean, the work to rebuild $C_{H_t}(u)$ for a vertex $u$ is proportional to the size of its lower neighborhood $N_u(1,i-1)$ and does not depend on $t$. Since each round colors a constant fraction of the vertices in $H_t$ in expectation, the expected number of rounds in which $u$ participates is $O(1)$, and the total palette-recomputation work per vertex is geometric.

\paragraph{Total work for coloring.} In section \ref{sec:coloring-procedure-complexity}, we show that the total expected work done for coloring all the marked and unmarked vertices is proportional to
\[\sum_{u \in \umrkd} |N_u(1,\ell(u))| + \sum_{u \in \mrkd} |N_u(1,\ell(u))|.\]

\subsubsection{Moving phase.}  Finally, we are left with the task of moving marked vertices. Recall that for each level $i$ we have $\upmrkd(i)$, the vertices on level $i$ that violate the Upper condition, and $\lmrkd(i)$, the vertices on level $i$ that satisfy the Upper condition but violate the Lower condition.

We handle these in two passes over the levels of the partition. In the first pass, we process the sets $\upmrkd(i)$ in a \emph{top-down} order using a raising procedure. Since raising a vertex from level $i$ only moves it to a higher level, this can never reduce the number of upper-neighbors of any vertex on a lower level $j<i$. This will be important in the final analysis as it will ensure that the initial set of upper marked vertices on each level stays the same. 

In the second pass, we process the sets $\lmrkd(i)$ in a \emph{bottom-up} order using a lowering procedure. Here, vertices only move downward, so processing level $i$ cannot reduce the number of down-neighbors of any vertex on a higher level. Again this has the important consequence that a vertex that is lower marked when we start the pass remains lower marked until we begin processing its level.

\paragraph{Raising procedure.} Let us consider how to process $\upmrkd(i)$. Recall the \textsc{Move} routine in Algorithm~\ref{alg:relaxed-seq-1}. If a vertex $u$ on some level $i$
violates the upper condition~\ref{con:upper:bchn}, the sequential algorithm scans levels from the top downward and moves $u$ to the highest level $k<\lambda$ such that
\[
  P(u,k) \;=\; \Ind{|N_u(1,k-1)| \ge \beta^{k-1} \;\land\; |N_u(1,k)| \le \beta^{k}} = 1.
\]
Our parallel raising procedure mimics this behavior, but works with batches of vertices rather than one vertex at a time. We process potential levels $k = \lambda, \lambda-1, \dots, i+1$ in order. In processing level $k$, we maintain the set of vertices on level $i$ that currently satisfy $P(u,k)=1$. Denote this set by $R_k$. Our goal is to process $R_k$ such that $\Theta(\beta^k |R_k|)$ tokens are released. 

It is tempting to raise $R_k$ in parallel to level $k$ and apply the same token argument we gave in the relaxed sequential algorithm (see \ref{desc:token-arg}) to conclude that this movement releases $O(\beta^k|R_k|)$ tokens. The problem is that the argument we gave there only applies in the case that $R_k$ forms an independent set. Indeed, our analysis crucially relied on the fact that after moving a vertex $u$ from level $i$ to level $k$ (where $k$ is such that $P(u,k) = 1$), the vertex token function for $u$, $\theta(u)$, is 0. To see why this is true, recall that on level $k$, $\theta(u) = \max(0,\beta^{k-1} - |N_u(1, k - 1)|)/2\beta$ and by definition of $k$, prior to moving $u$, $|N_u(1,k-1)| \geq \beta^{k-1}$. Further, after moving \textit{only} $u$, $|N_u(1,k-1)|$ remains the same. However, when we move all of $R_k$ in parallel, if any $u$'s neighbors also move, then $N_u(1,k-1)$ gets smaller, that is we can no longer guarantee that $|N_u(1,k-1)| \geq  \beta^{k-1}$. Instead, we can only say that 
\[|N_u(1,k-1)| \geq \beta^{k-1} - (\text{the internal degree of }u \in G[R_k]) = \beta^{k-1} - |N_u(i) \cap R_k|.\]
Thus, in this case $\theta(u)$ can increase by up to $|N_u(i) \cap R_k|/2\beta \leq \beta^{k-1}/2$. Therefore, a more careful analysis is required. Towards that end, lets suppose we have a uniform bound $\alpha$ on the internal degree some subset $M$ of $R_k$. We can show that by moving $M$, we release at least
\[\frac{|M|}{2}\left(\beta^{k-1} - \frac{\alpha(\beta+1)}{\beta} \right)\]
tokens. Thus, if we take $\alpha < \frac{\beta^{k-1}(\beta-2)}{\beta+1}$, the total decrease is $\beta^{k-1}|M| = \Theta(\beta^{k}|M|)$. Thus, if we can find an $M$ that, in addition to the bounded internal degree condition, has $|M| = \Theta(|R_k|)$ we would be done. We can achieve this with a symmetry breaking scheme as follows. Select every vertex in $R_k$ with probability $p$. In parallel, we move each vertex that was both selected and had at most $\alpha$ of its neighbors selected. With an appropriate choice of $p$, we show that $\E[M] = \Theta(|R_k|)$ and thus, this movement releases $\Theta(\beta^k |M|) = \Theta(\beta^k |R_k|)$ tokens in expectation. 

\paragraph{Lowering procedure.} Next, we describe how to process $\lmrkd(i)$. Recall that in the \textsc{Move} routine in Algorithm \ref{alg:relaxed-seq-1}, if a vertex $u$ violates the lower condition at level $i$, it is moved to level $i-4$. In the parallel setting, our goal is to move a large subset of $\lmrkd(i)$ down to level $i-4$ and in doing so release $\Theta(3^i |\lmrkd(i)|)$ tokens. As in the raising procedure, we achieve this using a symmetry breaking scheme: select each vertex in $\lmrkd(i)$ with probability $p$ and lower those that were both selected and had at most $\alpha$ of their neighbors selected. Call this subset $L$. We can show that in moving $L$, we release at least
\[|L|(\beta^{i-2} - 4\beta^{i-5} - 2\alpha)\]
tokens. Taking $\alpha < (\beta^{i-2} - \beta^{i-4} - 4\beta^{i-5})/2$, the total decrease is $\Theta(\beta^i |L|)$. Further, for an appropriate choice of $p$, we show that $\E[|L|] = \Theta(|\lmrkd(i)|)$ and thus overall moving $L$ to level $i-4$ releases $\Theta(3^i |\lmrkd(i)|)$ tokens in expectation.

\subsubsection{Full analysis.}
In what follows, we provide a sketch of the full analysis, eliding technical details for the sake of clarity. The high-level structure of the analysis in the parallel setting mirrors the sequential case.
As before, we must establish:
\begin{enumerate}[(1)]
    \item the cost of recoloring dirty vertices is dominated by the cost of moving those same vertices, and
    \item the movement of dirty vertices releases a number of tokens proportional to the work spent on their movement.
\end{enumerate}
Property~(2) is built into the design of the parallel raising and lowering procedures: by construction, every batch of raised (resp.\ lowered) vertices releases $\Theta(\beta^k |R_k|)$ (resp.\ $\Theta(\beta^i |\lmrkd(i)|)$) tokens, matching the $\Theta(\beta^k |R_k|)$ (resp.\ $\Theta(\beta^i \lmrkd(i))$) work needed for the corresponding updates to the data structures.

The main new difficulty is proving~(1).
In the sequential relaxed algorithm, every dirty vertex that is recolored is immediately moved, so its recoloring cost can be charged directly to its own movement.
In the parallel setting, this is no longer true: a dirty vertex may be recolored without being selected by the symmetry-breaking step in the raising or lowering procedure, and hence may not move in that batch.
We must therefore relate the cost of recoloring \emph{all} dirty vertices to the movement cost of only \emph{some} of them.

The lower-marked vertices are relatively easy to handle.
For vertices in $\lmrkd(i)$, the parallel lowering procedure moves a constant fraction of them in expectation, and each recoloring at level $i$ costs $\Theta(\beta^i)$.
Thus the total recoloring cost on level $i$ can be bounded by a constant factor times the total movement cost on the same level.
We therefore focus on the more subtle case of upper-marked vertices.

A key technical tool is the notion of a \emph{base level} associated with each timestamp in a vertex’s record.
If $\ell_i(u)$ is the level of $u$ at the time of timestamp $\tau_u^i$, we define
\[
  b_i(u) \;:=\; \max\Bigl\{\,\ell_i(u),\; \bigl\lfloor \log_\beta |N_u(1,\ell_i(u))| \bigr\rfloor \Bigr\}.
\]
The relaxed upper condition~\eqref{con:upper:new} ensures that when a vertex first becomes upper-marked, its base level is at least an additive constant (two) larger than its current level.

If an upper-marked vertex $u$ is selected in the raising procedure and moved, its recoloring cost at that timestamp is directly dominated by its own movement cost, as in the sequential analysis.
The difficult case is when $u$ is upper-marked but \emph{not} moved (really, the issue is when $u$ is not moved to a high enough level, but for simplicity, we will only consider this case). Here, the conditions for moving a vertex to a particular level guarantee that this can only happen if a large number of neighbors of $u$ participate in the same raising step and are actually moved up past $u$. More precisely, if $\tau_u^i$ is an important timestamp for $u$ at base level $b_i(u)$ and $u$ is not moved, then in that raising stage at least $\Omega(\beta^{b_i(u)})$ of $u$’s neighbors are moved upwards. The recoloring cost at $\tau_u^i$ is $\Theta(\beta^{b_i(u)})$, and we charge it to the $\Theta(\beta^{b_i(u)})$ total movement work contributed by those neighbors. This turns out to be sufficient.

Formally, we again decompose the work using the records of timestamps and the set of important timestamps $I$, split into clean and dirty parts $C_T$ and $D_T$, and then further into $CC_T$ and $DC_T$.
The same geometric-decay argument along recoloring chains shows that
\[
  W_T \;\le\; w(CC_T) + 2w(D_T),
\]
and a similar (albeit much more complicated) deferred-decision argument as in the sequential setting yields $\E[w(CC_T)] = O(T)$.
The new ingredient is showing that $w(D_T)$ is bounded by the total movement work, which in turn is covered by the token potential drop.
Combining these ingredients, we obtain
\[
  \E[W_T] \;=\; O(T \log \Delta)
\]
for the parallel algorithm as well.

%% file: sections/prelims.tex
\section{Preliminaries}

\subsection{Model and primitives.}  
\subsubsection{Model of computation.} In this paper, we use the work-span model with binary forking for analyzing parallel algorithms \cite{DBLP:journals/corr/abs-1903-04650}. As such, we assume we have a set of threads with access to a shared memory. Each thread supports the same operations as in the sequential RAM model, in addition to a fork instruction, which when executed, spawns two child threads and suspends the executing (parent) thread. When the two child threads end (through executing the end instruction), the parent thread starts again by first executing a join instruction. An individual thread can allocate a fixed amount of memory private to the allocating thread (refereed to as stack-allocated memory) or shared by all threads (referred to as heap-allocated memory). The \textit{work} of an algorithm is the total number of instructions carried out in the execution of the algorithm, and the \textit{span} is the length of the longest sequence of dependent instructions in the computation.

\subsubsection{Primitives.} We list some of the parallel primitives and state their bounds in the work/span model described above.

\begin{itemize}
    \item $\textsc{Scan}(A,\oplus)$ takes an array $A$ and an associative
    operator $\oplus$, and returns the sequence of all prefix sums of $A$
    under $\oplus$. A scan can be implemented using $O(|A|)$ work (assuming $\oplus$ takes $O(1)$ work) and $O(\log |A|)$ span.

    \item $\textsc{Reduce}(A,f)$ takes an array $A$ and an associative
    binary function $f$ and returns the reduction of all elements of $A$
    under $f$. A reduction can be computed with $O(|A|)$ (assuming $f$ is computable in $O(1)$ work) work and $O(\log |A|)$ span.

    \item $\textsc{Partition}(A,P)$ takes an array $A$ and a predicate $P$, and reorders $A$ so that all elements satisfying $P$ appear before those that do not. Using multiple scans, partitioning can be carried out with $O(|A|)$ work and $O(\log |A|)$ span.

    \item $\textsc{Filter}(A,P)$ takes an array $A$ and a predicate $P$
    and produces a new array consisting exactly of those $a \in A$ for
    which $P(a)$ holds, preserving their original order. Filtering can
    also be done in $O(|A|)$ work and $O(\log |A|)$ span.

    \item $\textsc{Semisort}(A,<)$ takes a sequence $A$ whose elements
    have keys (with respect to an ordering $<$) and permutes $A$ so that all entries with the same key appear contiguously. Semisorting can be done in $O(|A|)$ expected work and $O(\log |A|)$ span with high probability \cite{gu2015top}.
\end{itemize}

\subsection{The hierarchical partition.} Let $G = (V, E)$ be a simple undirected $\Delta$-bounded graph which undergoes a sequence of batch updates $\mc S := (S_1, S_2, \dots)$. We wish to maintain a proper $\Delta + 1$ coloring $\chi$ of $G$. Following \cite{Dynamic-Coloring-18}, we consider a hierarchical partition of $V$ that will prove useful in the design of our algorithm. 

Let $\lambda = \log_3 \Delta$. We will implicitly partition the vertex set $V$ into into $\lambda$ subsets $V_1, \dots, V_{\lambda}$. The level $\ell(u)$ of a vertex $u$ is the index of the subset it belongs to. For any vertex $u \in V$ and any two indices $1 \leq i \leq j \leq \lambda$, we let $N_u(i, j) := \{v : (u,v) \in E, \ i \leq \ell(v) \leq j\}$. Additionally for any level $i$ and any set $H$, we denote by $N_u(i)$ and $N_u[H]$ the set of $u$'s neighbors at level $i$ and the set of $u$'s neighbors in the induced subgraph on $H$ respectively. For each vertex in the hierarchical partition, we have the following two conditions of interest.

\begin{condition}[Lower]\label{con:lower}
    For every vertex $u \in V$ at level $\ell(u) > 1$, we have $|N_u(1, \ell(v) - 1)| \geq 3^{\ell(u) - 5}$.
\end{condition}

\begin{condition}[Upper]\label{con:upper}
    For every vertex $u \in V$, we have $|N_u(1, \ell(u))| \leq 3^{\ell(u) + 2}.$
\end{condition}
We say that a vertex is \emph{dirty} if it fails to satisfy either the lower or upper conditions. Otherwise, the vertex is said to be \emph{clean}.

Let $\mc C = \{0,\dots, \Delta\}$ be the set of possible colors. If $u$ is at level $i := \ell(u)$, let $\mathcal{C}_u^+ := \bigcup_{y \in N_u(i, \lambda)} \chi(y)$ and $\mathcal{C}_u^- = \mathcal{C} \setminus \mathcal{C}_u^+.$ We say a color $c \in \mathcal{C}^-_u$ is \emph{blank} for $u$ if no vertex in $N_u(5, i-1)$ is assigned color $c$. We say a color $c \in \mc{C}_u^-$ is \emph{unique} for $u$ if exactly one vertex in $N_u(5, i-1)$ is assigned color $c$. We let $B_u$ (respectively $U_u$) denote the blank (respectively unique) colors for $u$. Let $T_u := \mc C_u^- \setminus (B_u \cup U_u)$. The following two claims establishes useful bounds on the number of blank and unique colors in a vertex's palette. 

\begin{claim}\label{clm:many-unique-and-blank-colors}
    For any vertex $u$ at level $i$, we have $|B_u \cup U_u| \geq 1 + \frac{|N_u(1, i-1)|}{2}$.
\end{claim} 
\begin{proof}
    The claim follows from the following observations: $(1)$ $|\mc C_u^-| \geq 1 + |N_u(1, i-1)|$, $(2)$ $|\mc C_u^-| = |B_u \cup U_u| + |T_u|$, and $(3)$ $2|T_u| \leq |N_u(1, i - 1)|$.
\end{proof}

\begin{claim}\label{clm:many-blank-colors}
    Let $G_\text{blank}$ be a subgraph of $G$ in which each vertex is uncolored (blank). Then for $u \in V[G_\text{blank}]$, $|B_u| \geq \deg_{G_\text{blank}}(u) + 1$.
\end{claim}
\begin{proof}
    Consider $u \in \blank$. Define $N_u(c)$ to be the set of neighbors of $u$ with color $c$ and $t := |\{c : N_u(c) \not= \emptyset\}|$. Then $|B_u| = \deg(u) - t + 1$. Note the plus $1$ is due to the fact that $u$ itself is uncolored. Additionally, since each neighbor of $u$ in $G_{\text{blank}}$ is uncolored, $t \leq \deg(u) - \deg_{G_\text{blank}}(u)$. Thus $|B_u| = \deg(u) - t + 1 \geq \deg_{G_\text{blank}}(u) + 1$ as desired.  
\end{proof}

%% file: sections/static-coloring.tex
\section{Partial list-coloring}
We describe an important subroutine called partial list-coloring that will feature heavily in the coloring subroutines of the batch dynamic $(\deg + 1)$-coloring algorithm.

\subsection{Algorithm.}\label{sec:partial-coloring-alg}

The partial list-coloring procedure will partially color a given graph $H$ where each vertex $u$ has a color palette of size at least $\deg(u) + 1$. Here $H$ is represented by an array of vertices $V[H]$ and, for each $u \in V[H]$, the neighbor list of $u$ is stored in an array $N_u$. The color palette of each vertex $u$ is stored in an array $C_u$. 

The procedure is as follows. Every vertex $u$ samples a color $c$ uniformly at random from $C_u$ and tentatively sets $\chi(u) = c$. If any vertex in $N_u$ also sampled $c$, then the initial proposal is rejected, and we set $\chi(u) = -1$. All colored vertices are then deleted from $H$. This is done by partitioning both $V[H]$ and $N_u$ for all $u$. Then the latter part of $V[H]$ consists of the colored vertices. As a result of this procedure, we have a new uncolored graph $H'$ and a list $L$ of colored vertices from $H$.

\subsection{Analysis.}\label{sec:partial-coloring-analysis}

We show that the expected number of vertices that are successfully colored is a fraction of the original. By linearity of expectation, it suffices to lower bound the probability that a vertex $u$ is colored by a universal constant. Since each sample is chosen uniformly and independently, the probability can be written as
\begin{equation}\label{eqn:pcoloring-eq1}
    \sum_{c \in C_u} \Pr[\chi(u) = c] \left( \prod_{v \in N_u} 1 - \Pr[\chi(v) = c | \chi(u) = c] \right)
    = \frac{1}{|C_u|} \sum_{c \in C_u} \left( \prod_{v \in N_u} 1 - \frac{\mathbf{1}[c \in C_v]}{|C_v|} \right)
\end{equation}
Note that $\frac{\mathbf{1}[c \in C_v]}{|C_v|} \leq 1/2$ for each $v \in N_u$ as $|C_v| \geq \deg(v) + 1 \geq 2$. Then, using the fact that $1-x \geq \exp(-2x)$ for $x \leq 1/2$, we lower bound \ref{eqn:pcoloring-eq1} by
\begin{equation}\label{eqn:pcoloring-eq2}
    \frac{1}{|C_u|} \sum_{c \in C_u} \prod_{v \in N_u} \exp\left(-2\frac{\mathbf{1}[c \in C_v]}{|C_v|}\right) = \frac{1}{|C_u|} \sum_{c \in C_u} \exp\left(-2\sum_{v \in N_u} \frac{\mathbf{1}[c \in C_v]}{|C_v|}\right)
\end{equation}
Next, since $\exp(-2x)$ is convex, we can further lower bound \ref{eqn:pcoloring-eq2} with Jensen's inequality to get
\[\Pr[u \text{ is colored}] \geq \exp\left(-2\frac{1}{|C_u|} \sum_{c \in C_u} \sum_{v \in N_u} \frac{\mathbf{1}[c \in C_v]}{|C_v|} \right).\]
Finally, we observe that
\[
    \frac{1}{|C_u|} \sum_{c \in C_u} \sum_{v \in N_u} \frac{\mathbf{1}[c \in C_v]}{|C_v|}
        = \frac{1}{|C_u|} \sum_{v \in N_u} \sum_{c \in C_u} \frac{\mathbf{1}[c \in C_v]}{|C_v|} \leq \frac{1}{|C_u|} \sum_{v \in N_u} \frac{|C_u \cap C_v|}{|C_v|} \leq \frac{1}{|C_u|} \sum_{v \in N_u} 1 < 1.
\]
Thus, we have that $\Pr[u \text{ is colored}] \geq \exp(-2)$. As such, the expected number of colored vertices is at least $|V[H]|/\exp(2)$. We summarize this by the following Lemma.

\begin{lemma}\label{lem:partial-list-coloring-number-colored}
   The partial $(\deg + 1)$-coloring algorithm colors at least $|V[H]|/\exp(2)$ vertices in expectation.
\end{lemma}

We also analyze the work and span of the partial $(\deg + 1)$-coloring algorithm. To determine whether the tentative color $c$ is unique for $u$, we simply scan (in parallel) over $N_u$. Additionally, updating the graph requires partitioning $V[H]$ and $N_u$ for each $u$ based on the predicate $\Ind{\chi(u) = -1}$. Therefore, the total work and span is $O(|E[H]|)$ and $O(\log |V[H]|)$ where $E[H]$ is the edge set of $H$. We summarize this in the following Lemma.

\begin{lemma}\label{lem:partial-list-coloring-work-span}
The partial $(\deg + 1)$-coloring algorithm uses $O(|E[H]|)$ work and $O(\log |V[H]|)$ span. 
\end{lemma}

 \subsection{Static list coloring}\label{sec:static-list-coloring} Using the partial-coloring subroutine we can design a static $(\deg + 1)$-coloring algorithm. We are given an initial graph $H = (V,E)$ with a color palette $\mc C_u$ for each $u \in V$ such that $|\mc C_u| = \deg(u) + 1$. The algorithm iterates the partial list-coloring routine, updating the coloring instance by setting $G = H'$ and removing from $\mc C_u$ all of the colors used by $N_u \cap L$, until each vertex is successfully colored.

 \subsubsection{Work and span analysis}
 If we denote by $E_i$ the number of edges remaining at the start of round $i$ and $R$ the total number of rounds, then the work and span are easily seen to be proportional to 
\[W = \sum_{i=1}^R E_i \quad \text{and} \quad D = R \log n\]
by Lemma \ref{lem:partial-list-coloring-work-span}.
To bound the above quantities, we establish the following useful inequality
\begin{equation}\label{eqn:edge-ineq}\E[E_i] \leq \lambda^i |E| \end{equation}
where $\lambda := 1 - \exp(-2)$. Observe that the probability that an edge is deleted is at least the probability that one of its adjacent vertices is deleted, which in turn is at least $\exp(-2)$. In each round, we therefore expect at most a $\lambda$ fraction of the current edges to remain, so $\E[E_i \mid E_{i-1}] \leq \lambda E_{i-1}$ for all $i$. Finally,
\[\E[E_i] = \E[\E[E_i \mid E_{i-1}]] \leq \lambda \E[E_{i-1}] \text{ for all } i.\]
We conclude \ref{eqn:edge-ineq}.  With \ref{eqn:edge-ineq}, we can bound the expected work of the algorithm as
\[\E[W] = \E\left[\sum_{i=1}^\infty E_i\right] = \sum_{i = 0}^{\infty} \E[E_i] \leq \sum_{i=0}^{\infty} \lambda^i|E| = \left(\frac{1}{1-\lambda}\right)|E| = O(|E|).\]
Next, we bound the depth of the algorithm. Let $t = 2\log_{1/\lambda} |E|$ and $c \geq 1$. By Markov's inequality and \ref{eqn:edge-ineq} we have
\[\Pr[R > ct] = \Pr[E_{ct} > 1] \leq \E[E_{ct}] \leq \lambda^{ct} |E| < 1/|E|^c.\]
Thus, the algorithm terminates \whp \ after $O(t)$ rounds and thus has $O(t\log n) = O(\log^2 n)$ span \whp. We restate this in the following Lemma.

\begin{lemma}\label{lem:static-list-coloring-complexity}
    We can $(\deg + 1)$ color a graph $H = (V,E)$ using $O(|E|)$ expected work and $O(\log^2 n)$ span \whp. 
\end{lemma}

%% file: sections/implementation.tex
\section{Data structures and update framework} In this Section, we detail the data structures required by our algorithm along with a framework for applying batches of updates in parallel.

For each vertex $u$ we maintain the following data.
\begin{itemize}
    \item for $\ell(u) \leq i \leq \lambda$, $N_u(i)$: neighbors of $u$ in level $i$.
    \item $N_u(5, \ell(u) - 1)$ : neighbors of $u$ at levels below $\ell(u)$.
    \item$\mc C_u^+, \mc C_u^-$: the upper and lower color palettes respectively.
    \item $\tau_u := (i, \ell(u))$: the timestamp tuple which stores the batch number and level of $u$ at the time when $u$ was last recolored randomly. If $u$ was colored deterministically, $\tau_u := \text{det}$. Note we say $\text{det} > (i, \ell(u))$ for any $i$ and $\ell(u)$, i.e. we always prioritize recoloring vertices that were previously colored deterministically. 
    \item for each color $c$, $\mu_u^+(c)$: a counter such that if $c \in \mc C_u^+$, then $\mu_u^+(c)$ equals the number of neighbors in $N_u(\ell(u), \lambda)$ with color $c$. Otherwise, $\mu_u^+(c)$ is $0$.
    \item $\ell(u)$: the level of $u$.
\end{itemize}

Additionally, we maintain the following global data.
\begin{itemize}
    \item $\chi$ : the coloring of $G$.
    \item for $5 \leq i \leq \lambda$, $\upmrkd(i)$ : set of vertices on level $i$ which are blank and don't satisfy the upper condition \ref{con:upper}.
    \item for $5 \leq i \leq \lambda$, $\lmrkd(i)$ : set of vertices on level $i$ which are blank and only don't satisfy the lower condition \ref{con:lower}.
    \item for $5 \leq i \leq \lambda$, $\umrkd(i)$ : set of vertices on level $i$ which are both blank and clean.
\end{itemize}
In the remainder of the section, we list the data structures we require.

\subsection{Data structures.} Here we list the main data structures used by our algorithm. Table \ref{tab:ds-notation} list the data structures and notation used to represent each piece of data.

\begin{table}[ht]
\centering
\renewcommand{\arraystretch}{1.1}
\resizebox{\linewidth}{!}{%
\begin{tabular}{@{}l l l l@{}}
\toprule
\textbf{Object} & \textbf{Meaning} & \textbf{Data structure} & \textbf{Notation} \\
\midrule
$N_u(i)$ & Neighbors of $u$ on level $i$ & Batch-parallel hash table & $\mathsf{NbrHT}[u,i]$ \\
$N_u(5,\ell(u)-1)$ & Neighbors of $u$ below $\ell(u)$ & Batch-parallel hash table & $\mathsf{BelowHT}[u]$ \\
$\mc C_u^-$ & Lower color palette slice & Dynamic partitioned array (slice) & $A_u[0:s_u-1]$ \\
$\mc C_u^+$ & Upper color palette slice & Dynamic partitioned array (slice) & $A_u[s_u:\deg(u)]$ \\
$\mu_u^+(c)$ & \# upper-level neighbors of color $c$ & Integer& $\mu_u^+(c)$ \\
$\tau_u$ & Recolor timestamp or $\text{det}$ & Tuple & $\tau_u \in\{(i,\ell),\mathsf{DET}\}$ \\
$\ell(u)$ & Level of $u$ & Integer scalar & $\ell(u)$ \\
\midrule
$\chi$ & Coloring of $G$ & Integer array of size $|V|$ & $\chi$ \\
$\upmrkd(i)$ & Blank vertices failing upper cond.\ at level $i$ & Batch-parallel hash table & $\upmrkd(i)$ \\
$\lmrkd(i)$ & Blank vertices failing only lower cond.\ at level $i$ & Batch-parallel hash table & $\lmrkd(i)$ \\
$\umrkd(i)$ & Blank \& clean vertices at level $i$ & Batch-parallel hash table & $\umrkd(i)$ \\
\bottomrule
\end{tabular}
}
\caption{Data items, their representations, and notation. When there is no risk of confusion, we use the same notation. The palette container is $(A_u,M_u,s_u)$ where $M_u$ maps colors to positions; $s_u$ is the cut index. Treat $\mathsf{DET}$ as a sentinel greater than any pair $(i,\ell)$.}
\label{tab:ds-notation}
\end{table}

\paragraph{Hash table.} We use a batch-parallel hash table to represent the majority of the data. Let $B$ be a batch of keys stored in an array. Using the implementation of \cite{Gil91a}, we can support the following operations.

\begin{itemize}
    \item $\textsc{Insert}(B)$ and $\textsc{Delete}(B)$: These operations require $O(|B|)$ expected work and $O(\log |B|)$ span \whp. 
    \item $\textsc{List}()$. This operation returns each key in the table in the form of an array using $O(\text{table size})$ work and $O(1)$ span.
    \item $\textsc{Clear}()$. This operation deletes the table and requires $O(\text{table size})$ work and $O(1)$ span. 
\end{itemize}

\paragraph{Dynamic partitioned array.} To represent $\mc C_u^-$ and $\mc C_u^+$, we use a dynamic partitioned array which consist of an element array $A_u$, a map array $M_u$, and a split index $s_u$. Then $\mc C_u^- = A_u[0:s_u - 1]$ and $\mc C_u^+ = A_u[s_u : \Delta]$. Further, the position of $c \in [\deg(u)]$ in $A_u$ is given by $M_u[c]$. Note that $|A_u| = |M_u|$. Next, we describe the operations supported and how to implement them. 

Let $B$ be a batch of colors $c$ where $c \in \mc C_u$. We assume $B$ is stored in an array. The dynamic partitioned array supports the following operations.

\begin{itemize}
    \item \textsc{MoveUp}$(B)$: First, we filter out elements of $B$ which already reside in $\mc C_u^+$. Next, using a semisort, we deduplicate $B$. At this point we can assume $B \subseteq \mc C_u^-$. The operation moves each color in $B$ to $\mc C_u^+$. To do this, we first rearrange $A[0:s_u-1]$ so that $B$ is at the end of the slice. Define $P := A_u[s_u-|B|-1:s_u-1]$, so $P$ is the $|B|$ sized suffix of $\mc C_u^-$. Using the map array $M_u$ we determine which colors from $B$ are in $P$ and use a parallel partition to place them at the end of $P$. Let $B'$ be an array of the remaining colors in $B$. For each $c = B'[i]$, swap $c$ with the $i$th element of $P$, that is swap $A_u[M[c]]$ with $P[i]$. Throughout this process, we  additionally maintain $M_u$ by swapping the appropriate values. Finally, we decrement $s_u$ by $|B|$. Its clear that this operation requires $O(|B|)$ expected work and $O(1)$ span.
    \item \textsc{MoveDown}$(B)$: This operation moves each color in $B$ to $\mc C_u^-$. To do this, we rearrange $A[s_u:\Delta]$ so that $B$ is at the start of the slice and increment $s_u$ by $|B|$. The details of the operation are entirely analogous to the \textsc{MoveUp} procedure. Thus we deduce that the \textsc{MoveDown} operation requires $O(|B|)$ expected work and $O(1)$ span.
    \item $\textsc{RearrangeLower}(B)$: This operation moves each color in $B$ to the back of $\mc C_u^-$. This is done using the same algorithm for $\textsc{MoveUp}$ without the decrement to $s_u$.
\end{itemize}

Finally, we will represent the coloring $\chi$ with a simple array of size $|V|$. We note that if $u \in V_\text{blank}$ then $\chi(u) = -1$.

\subsection{Update framework.} During the course of our algorithms, we will need to make numerous updates to the data structures in parallel. To avoid issues of concurrency, we will synchronize these updates. To do this, during the course of an algorithm, when we need to make an update (e.g. due to a vertex changing levels), we will make a record of that update in an array associated to the vertex which caused the update to happen. That is, each vertex $u$ maintains an array $R_u$ where update requests are recorded. Then at the next synchronous point, we group these update requests by target and apply them as batch updates for the appropriate data structures. Formally, for each update we define an \emph{update request} which is defined as a tuple $(t,\; \text{tag},\; op, \;args)$ where $t$ is the target of the request, tag specifies which data structure of $t$ to update, and $op$ and $args$ specifies the operation to be performed and its corresponding arguments. The target $t$ can either be a single vertex or the global state, that is $t = v$ for some $v \in V$ or $t = G$, to denote the target as a local or global data structure of the graph. As an example, consider the update request $(v,\; \mathsf{Nbr}_i,\;\textsc{Insert},\; u)$. This request, made by $u$, adds $u$ to $v$'s level $i$ neighborhood list, which is given by the tag $\mathsf{Nbr}_i$.

\paragraph{Constructing request arrays.} In our algorithm, all updates are prompted by some pair $(u,v)$ where $v$ is in some restricted subset $N$ of $N_u$. What $N$ is, which pairs cause updates, and the kinds of update they cause will be described in detail in all the contexts we consider. But in every scenario, each pair causes at most a constant number $c$ of update requests. Thus to create $R_u$, we first allocate an array $R_u'$ of size $c |N|$. Then in parallel for each pair $(u,v)$ with $v \in N$, we put the updates it creates in the corresponding size $c$ segment of $R_u'$. Then we filter $R_u'$ to obtain the array $R_u$ consisting of all updates caused by pairs $(u,v)$. We summarize this in the following Lemma.

\begin{lemma}\label{lem:request-record-construction}
    Constructing $R_u$ requires $O(|N|)$ work and $O(\log n)$ span.
\end{lemma}

\paragraph{Applying multiple update requests.} Let $U$ be a set of vertices which have nonempty update request arrays. We form the global array $R := R_0R_1 \cdots R_{|U|-1}$ using a parallel prefix on sizes, in $O(|R|)$ work $O(\log n)$ span. We then apply a parallel semisort of $R$ by the pair $(t, \text{tag})$, producing groups
\[
    G_{t, \text{tag}} = \left((t, \text{tag}), \langle(op_1, args_1),\dots,(op_t,args_t) \rangle \right).
\]

Each group corresponds to a \emph{single} target $t$ and a \emph{single} data structure (identified by tag). We process all groups independently and in parallel. Within a fixed $G_{t, \text{tag}}$, we further group the sequence $\langle(op_1, args_1),\dots,(op_t, args_t) \rangle\rangle$ by type of operation to produce the batches
\[
    B_{t, \text{tag}}^{op} := [args_j : (op_j, args_j) \in G_{t,\text{tag}}, op_j = op].
\]
Each batch $B_{t, \text{tag}}^{op}$ is exactly the input to the corresponding batch-parallel primitive $op$ of that data structure. To finish, in parallel over each $G_{t,\text{tag}}$, we call, in sequence, the operation $op(B_{t,\text{tag}}^{op})$ on the corresponding data structure.


Next, we analyze the complexity of this procedure.

\begin{lemma}\label{lem:update-framework-time}
    Let $R = \sum_{t \in |T|} |R_t|$ be the total number of update requests. The procedure above applies all requests using $O(R)$ expected work and $O(\log n)$ span \whp.
\end{lemma}

\begin{proof}
    Computing the batches $B^{op}_{t, \text{tag}}$ requires two semisorts and a prefix sum over $R$ elements. Each of these steps requires $O(R)$ expected work and $O(\log R)$ span \whp. Additionally, its not hard to verify that for each operation $op$ that is supported by our data structures, $op(B_{t,\text{tag}}^{op})$ requires $O(|B_{t,\text{tag}}^{op}|)$ expected work and $O(\log |B_{t,\text{tag}}^{op}|)$ span \whp. Since the batches form a partition of the requests, we conclude that the total expected work of applying $R$ update requests is $O(R)$. To prove the span,  observe that $R = O(n^2)$ and thus by the above analysis, computing the batches requires $O(\log n)$ span \whp. Next, for any $B_{t,\text{tag}}^{op}$ fix $b = |B_{t,\text{tag}}^{op}|$ and let $X$ denote the parallel span for applying $op(B_{t,\text{tag}}^{op})$. We have that for any sufficiently large constant $c_1$,
    \[\Pr[X > c_1 \log b] \leq \frac{1}{b^{c_1}}.\]
    In particular, for any $c_2$, taking $c_1 = c_2\frac{\log n}{\log b}$ gives that
    \[\Pr[X > c_2 \log n] \leq \frac{1}{n^{c_2}}.\]
    Taking $c_2$ sufficiently large and applying a union bound then gives that \whp, every operation finishes in $O(\log n)$ span.
\end{proof}

%% file: sections/algorithm1.tex
\section{Parallel dynamic $(\Delta + 1)$-coloring}\label{sec:algorithm}
In this section we instantiate the high-level algorithm from Section~\ref{sec:tech-sec-3} and prove Theorem~\ref{thm:main-thm}. Given a batch $S$ of edge updates, our update routine proceeds in three phases:
initialization, coloring, and moving.

In the initialization phase (Section~\ref{sec:init-phase}) we apply $S$ to the graph, update all neighborhood and palette data structures, and identify the set $V_{\blank}$ of vertices that must be recolored. Each $u\in V_{\blank}$ is then classified by its level and cleanliness into one of three
sets: upper-marked, lower-marked, or unmarked, giving the tables $\upmrkd(i)$, $\lmrkd(i)$, and $\umrkd(i)$. This phase uses $O(|S|)$ expected work and $O(\log n)$ span \whp.

In the coloring phase (Section~\ref{sec:coloring-phase}) we restore a proper $(\Delta+1)$-coloring by recoloring $V_{\blank}$ level-by-level from the top down. For each level $i$ we first color the marked vertices $\mrkd(i)$ using only blank colors, and then color the unmarked vertices $\umrkd(i)$ using palettes of blank-or-unique colors, both via repeated applications of the partial list-coloring subroutine from Section~\ref{sec:partial-coloring-alg}. We show that this phase has $O(\log^2 n)$ span per level whp, and that its total expected work is proportional to $\sum_u |N_u(1,\ell(u))|$ over all vertices involved in coloring.

In the moving phase (Section~\ref{sec:moving-phase}) we move marked vertices to appropriate levels in order to release a sufficient number of tokens. We first perform a top-down \emph{raising} pass on the upper-marked sets $\upmrkd(i)$. We then perform a bottom-up \emph{lowering} pass on the lower-marked sets $\lmrkd(i)$, moving vertices down four levels at a time. 

Finally, in Section~\ref{sec:full-analysis} we relate the work of recoloring marked vertices to the movement work, and use a deferred decision argument to analyze the work recoloring clean vertices. With a token analysis, we show that the total movement work over any sequence of batches is $O(T\log\Delta)$ in expectation. Putting everything together, we complete the proof of Theorem~\ref{thm:main-thm}.

\subsection{Initialization phase}\label{sec:init-phase} We describe how to apply a batch of edge updates $S$ and compute the initial tables for the sets $\upmrkd(i), 
\lmrkd(i),$ and $\umrkd(i)$. Let $e = \{u,v\}$ be in $S$. We handle insertions and deletions separately.

We handle insertions first. To start, we apply the structural updates to the graph. As such, for every vertex $u$ incident to an edge $e = \{u,v\}$ being inserted, we make a request to add $u$ to $\nbr{v}{\ell(u)}$, provided $\ell(u) \geq \ell(v)$, and $\belowHT{v}$ otherwise. These constitute the structural updates. What remains are color updates. For each $u$, we identify if it is part of a monochromatic edge $\{u,v\}$ being inserted where $u$ has a later timestamp than $v$. If so, we add it to the appropriate set. In particular, if $u$ is dirty with respect to the upper condition, we make a request to add it to $\upmrkd(\ell(u))$. Otherwise, if $u$ is only dirty with respect to the lower condition, we make a request to add it to $\lmrkd(\ell(u))$. Finally, if $u$ is clean we add it to $\umrkd(\ell(u))$. Now these vertices should be blank, so we will also remove their colors. Accordingly, we make a request to decrement the color counters of $\chi(u)$ for each $v$ in $N_u(1,\ell(u))$ and set $\chi(u) = -1$. After applying these updates, if any of the color counters $\mu_v^+(\chi(u))$ move to zero, we request updates to move $\chi(u)$ down in $\mc C_v$. For all non-blank vertices $u$, and each $\{u,v\}$ added so that $\ell(u) \geq \ell(v)$, we request to move $\chi(u)$ up in $\mc C_v$ and increment $\mu_u^+(\chi(u))$.

Now, we handle deletions. Again, we first apply structural updates. For each vertex $u$ incident to a deleted edge $e = \{u,v\}$, we request to delete $u$ from $\nbr{v}{\ell(u)}$ if $\ell(u) \geq \ell(v)$ and from $\belowHT{v}$ otherwise. The structural updates are complete so we move on to color updates. For each $v \in N_u(1,\ell(u))$, we make a request to decrement the color counter $\mu_v^+(\chi(u))$. After applying these updates, if any of the color counters $\mu_v^+(\chi(u))$ move to zero, we request updates to move $\chi(u)$ down in $\mc C_v$. This concludes the initialization phase.

The work and span of the initialization phase is summarized in the following Lemma.
\begin{lemma}\label{lem:init-work-span}
    The initialization phase uses $O(|S|)$ expected work and $O(\log n)$ span \whp.
\end{lemma}

\subsection{Coloring phase}\label{sec:coloring-phase} We color the set of blank vertices in a top-down fashion processing one level at a time. To color $V_\blank(i)$, we consider the marked and unmarked vertices separately as the color palettes we use in each case differ significantly.

To handle each case, we will use procedures of the following form. First, we will initialize a coloring instance $(H_0,C_{H_0})$ comprising an uncolored subgraph and a color palette for the subgraph. Then given a coloring instance $(H_i,C_{H_i})$, we will apply the partial coloring routine given in Section \ref{sec:partial-coloring-alg} to get a graph $H'$ of the the uncolored vertices and $L$ a list of the colored vertices. After this routine, the coloring $\chi$ will be defined on the vertices from $L$. Using $L$ and the assigned coloring $\chi$ on $L$, we will make updates (using the framework of the previous section) to the necessary data structures. After making these updates, we will prepare a new coloring instance $(H_{i+1}, C_{H_{i+1}})$ where $H_{i+1} = H'$ and $C_{H_{i+1}}$ is chosen appropriately. This procedure is iterated until all vertices from the initial subgraph $H_0$ are colored. 

We remark that for these routines to achieve logarithmic span \whp, we need to ensure that at each iteration $i$, a fraction of the vertices $H_i$ are expected to be colored. By the analysis in Section \ref{sec:partial-coloring-analysis} it suffices to have $|C_{H_i}(u)| \geq \deg_{H_i}(u)+1$. We will ensure that every coloring instance satisfies this constraint.

\subsubsection{Marked coloring procedure.}\label{sec:mrkd-coloring}

In the marked case, our initial coloring instance $(H_0, C_{H_0})$ has $H_0 = G[\mrkd(i)]$. To obtain $C_{H_0}$, we first preprocess the dynamically maintained palettes $\mc C_u$ so that the first $\deg_{H_0}(u)+1$ entries of $\mc C_u^-$ are blank colors for $u$. We do this by first computing the set $NB$ of non-blank colors which is done by copying the color of each vertex in $\belowHT{u}$ to an array and then deduplicating using a semisort. Then $NB$ is moved to the end of $\mc C_u^-$ using the $\textsc{RearrangeLower}$ operation. Since the colors in $NB$ comprise all of the non-blank colors appearing in the lower palette, all of the remaining colors are blank and at the front of the lower palette. We are guaranteed that the number of blank colors is at least $\deg_{H_0}(u)+1$ by Claim \ref{clm:many-blank-colors}. For each $u \in H_0$, we let $C_{H_0}(u)$ consist of the first $\deg_{H_0}(u)+1$ colors in $\mc C_u^-$. We apply the partial list coloring, getting $H'$ and $L$, and set $H_1 := H'$. Next, we make the required data structure updates using $L$. We detail these updates in Section \ref{sec:coloring-procedure-updates}. To compute $C_{H_1}$, for each remaining vertex $u$, we update $C_{H_0}(u)$ to remove the colors of neighbors in $L$. We now apply the partial list coloring algorithm with input $(H_1,C_{H_1})$. We do this repeatedly until $\mrkd(i)$ is empty. The pseudocode for this procedure is given in Algorithm \ref{alg:mrkd-coloring}.

\begin{algorithm}[ht]
\function{\emph{\textsc{ColorMarked}$(\mrkd(i))$}}{
    Set $H = G[\mrkd(i)]$. \\
    \parfor{$u \in \mrkd(i)$}{
        Compute the set of non-blank colors $NB$. \\
        \textsc{RearrangeLower}$(NB)$. \\
        Set $C_{H}(u) = \mc C_u^-[0:\deg_{H}(u)]$.
    }
    \While{$H$ is not empty}{
    $(H',L) \gets \textsc{PartialColor}(H,C_{H})$. \\
    $H \gets H'$.\\
    Update data structures using $L$.*\\
    \parfor{$u \in V[H]$}{
        Set $C_{H}(u) = \mc C_H(u) \setminus \{\chi(v) : v \in N_u(L)\}$.
    }
    }
}
\caption{Marked Coloring Procedure. \hfill *See Section \ref{sec:coloring-procedure-updates}.}
\label{alg:mrkd-coloring}
\end{algorithm}

\subsubsection{Unmarked coloring procedure.}

In the unmarked case, our initial coloring instance $(H_0, C_{H_0})$ has $H_0 = G[\umrkd(i)]$. To obtain $C_{H_0}$, we preprocess the dynamically maintained palettes $\mc C_u$ so that the first $\max(\deg_{H_0}(u)+1, 3^{i-5}/2 + 1)$ entries of $\mc C_u^-$ are blank or unique colors for $u$. As in the Marked procedure, we do this by first computing the set $NU$ of non-unique and non-blank colors using a semisort to group the array of colors used by vertices in $\belowHT{u}$. Then $NU$ is moved to the end of $\mc C_u^-$ using the \textsc{RearrangeLower} operation. Since these colors comprise all of the colors in the lower palette which are not blank or unique, the colors which are blank or unique now occupy a prefix of the lower palette. We are guaranteed that the number of blank or unique colors is at least $\max(\deg_{H_0}(u)+1, 3^{i-5}/2 + 1)$ by Claims \ref{clm:many-unique-and-blank-colors} and \ref{clm:many-blank-colors}. For each $u \in H_0$, we let $C_{H_0}(u)$ consist of the first $\max(\deg_{H_0}(u)+1, 3^{i-5}/2+1)$ colors in $\mc C_u^-$. We apply the partial list coloring, getting $H'$ and $L$, and set $H_1 := H'$. We make data structure updates using $L$ which are described in Section \ref{sec:coloring-procedure-updates}. We note that in doing this, for all colors chosen that were unique, we add the corresponding vertices to the marked or unmarked vertices on the appropriate level. Unlike in the marked case where we can compute the palettes for the next instance using $C_{H_0}$ and $L$, in the unmarked case we still need to ensure that $|C_{H_1}(u)| \geq 3^{i-5}/2 + 1)$ for each $u \in V[H_1]$. To do this, we simply compute the new color palettes in the same way as the initial color palettes were computed. This ensures the proper guarantees are met. We repeat this process until $\umrkd(i)$ is empty. The pseudocode for this procedure is given in Algorithm \ref{alg:umrkd-coloring}.

\begin{algorithm}[ht]
\function{\emph{\textsc{ColorUnmarked}$(\umrkd(i))$}}{
    Set $H = G[\umrkd(i)]$. \\
    \While{$H$ is not empty}{
    \parfor{$u \in \umrkd(i)$}{
        Compute the set of non-unique non-blank colors $NU$. \\
        $\textsc{RearrangeLower}(NU)$. \\
        Set $C_{H}(u) = \mc C_u^-[0:\max(\deg_{H}(u), 3^{i-5}/2 + 1)]$.
    }
    $(H',L) \gets \textsc{PartialColor}(H,C_{H})$. \\
    $H \gets H'$.\\
    Update data structures using $L$.*
    }
}
\caption{Unmarked Coloring Procedure. \hfill *See Section \ref{sec:coloring-procedure-updates}.}
\label{alg:umrkd-coloring}
\end{algorithm}

\subsubsection{Coloring procedure complexity analysis.}\label{sec:coloring-procedure-complexity} We will show that the expected work for both the marked and unmarked coloring procedures is proportional to the sum over $|N_u(1,\ell(u))|$ for each vertex $u$ involved.



\begin{lemma}\label{lem:coloring-work-span}
Coloring each level takes $O(\log^2 n)$ span \whp. The expected work performed for the coloring phase is proportional to the sum of $|N_u(1,\ell(u)|$ over every vertex $u$ involved in the coloring procedure.
\end{lemma}
\begin{proof}
    The cost of the updates in both procedures is exactly as desired due to Lemma \ref{lem:coloring-procedure-updates}, so we may ignore them for the rest of the analysis.
    
    With this under consideration, the marked coloring procedure is equivalent to the static list coloring, which we showed to have work proportional to the size of the initial subgraph in Section \ref{sec:static-list-coloring}. Thus in total, the marked coloring procedure costs $O(\sum_{i=1}^\lambda |E[\mrkd(i)]|)$ which is upper bounded by the sum of $|N_u(1,\ell(u)|$ over every vertex $u$ involved in the coloring procedure.
    
    To understand the unmarked coloring procedure, fix a level $i$. As unmarked vertices are clean, then $|N_u(1,\ell(u)| = \Theta(\beta^{\ell(u)})$ for each $u \in \umrkd(i)$. Thus, the cost of computing the palettes in each round $t$ is $O(|V[H_t]| \beta^{\ell(u)})$. By Lemma \ref{lem:partial-list-coloring-work-span}, we get that the cost of coloring the $i$th level of unmarked vertices is upper bounded by
    \[
        \sum_{t} |E[H_t| + |V[H_t]|\cdot \beta^{\ell(u)} = O\left(\sum_{t} |V[H_t]|\cdot \beta^{\ell(u)}\right).
    \]
    Since the vertex set of the coloring instance decreases by a constant in expectation, we see that the expected cost is
    \[
        O(|\umrkd(i)| \cdot \beta^{\ell(u)}) = O\left(\sum_{u \in \umrkd(i)} \beta^{\ell(u)}\right).
    \]
    Summing over $i$, this is of course upper bounded by the sum of $|N_u(1,\ell(u))|$ over every vertex $u$ involved in the coloring procedure.
\end{proof}

\subsubsection{Coloring procedure updates.}\label{sec:coloring-procedure-updates} We detail the updates made for a single iteration of Algorithm \ref{alg:mrkd-coloring} and \ref{alg:umrkd-coloring}, i.e. for a single iteration of the while loop.

Fix $u \in L$ and denote by $c$ the color assigned to $u$. In both algorithms, for each $v \in N_u(1,i)$, we make an update request to increment $\mu_v^+(c)$ and move up $c$ in $\mc C_v$. Now, if $u$ was unmarked, we additionally check if $c$ is a unique color by parallel looping over $N_u(1,i-1)$. If $c$ is not unique, then $u$ makes no further update requests. So suppose that $c$ is unique, and thus $u$ has a corresponding unique lower neighbor $v$ with color $c$. At this point, there can be multiple vertices in $\umrkd(i)$ that have $v$ as their unique neighbor. To handle this, we will have each such vertex make a request to add $v$ to a temporary hash table $A$. After applying these requests, the table $A$ consists of every vertex in the lower neighborhood of some vertex in $\umrkd(i)$ which now needs to be uncolored. For each $v \in A$, let $c$ denote the current color of $v$. We first uncolor $v$ by setting $\chi(v) = -1$. Then, for each $w \in N_v(1,\ell(v))$, we make an update request to decrement $\mu_w^+(c)$. After applying the updates, if $\mu_w^+(c) = 0$, then we make a request to move down $c$ in $\mc C_w$. Finally, if $w$ is dirty with respect to the upper condition, we make a request to add it to $\upmrkd(\ell(w))$. Otherwise, if $w$ is only dirty with respect to the lower condition, we make a request to add it to $\lmrkd(\ell(w))$. Finally, if $w$ is clean we add it to $\umrkd(\ell(w))$. After applying these updates, we set the timestamps of each $u \in L$. If $u \in \mrkd(i)$ we set $u$'s timestamp $\tau_u$ to be $\det$. Otherwise, if we are currently processing batch $k$, we set $\tau_u = (k,i)$. 

The complexity of these updates is summarized in the following lemma.
\begin{lemma}\label{lem:coloring-procedure-updates}
    Every round of updates takes $O(\log n)$ span \whp. The expected work performed in executing the updates is proportional to the sum of $|N_u(1,\ell(u)|$ over every vertex $u$ involved in the coloring procedure.
\end{lemma}

%
%

\subsection{Moving phase}\label{sec:moving-phase}

In this section, we detail the raising and lowering procedures, and give the formal analysis of the amortized cost of doing so. As the amortized analysis dictates the design of these procedures, we now give the token functions which we will use to argue the amortized cost. These token functions are very similar to the ones used in the dynamic $O(\log \Delta)$ algorithm of \cite{Dynamic-Coloring-18}, though our analysis diverges significantly. For every edge $(u,v)$ and vertex $v$, we associate the token functions
\begin{equation}\label{eqn:edge-potential}
        \theta(u, v) = \lambda - \max(\ell(u), \ell(v)).
\end{equation}   
\begin{equation}\label{eqn:vertex-potential}
    \theta(v) =
    \begin{cases} \frac{\max(0, 3^{\ell(v) - 1} - |N_v(1, \ell(v) - 1)|)}{6} & \text{if } \ell(v) > 5; \\
            0 & \text{otherwise}.
    \end{cases}
\end{equation}
Every insertion of an edge $(u,v)$ increases the total number of tokens by $\lambda-\max(\ell(u), \ell(v)) \leq \lambda$. Moreover a deletion of an edge $(u,v)$ removes the tokens associated to that edge and increases the token count of each endpoint by at most $1$, so every deletion increases the token count by at most $2 \leq \lambda$. Each token will represent a constant amount of computational power. 

Our strategy to show an amortized cost of $O(\lambda) = O(\log \Delta)$ is to track the token decrease caused by each movement procedure, and show it is proportional to the cost of executing that movement procedure (in expectation). Doing so will imply that all movements in aggregate cost an amortized $O(\lambda)$ work per edge update in expectation.

\subsubsection{Raising procedure description.}\label{sec:raising-procedure}
We describe a procedure to raise vertices in $\upmrkd(i)$. The pseudocode for this procedure is given in Algorithm \ref{alg:Raise}. In the remainder of the section, we describe Algorithm \ref{alg:Raise} and show that its expected work is bounded by
\[ 
    O\left(\lambda |\upmrkd(i)| + \sum_{k=i+1}^\lambda 3^k \E[|R_k|]\right)
\]
where $R_k$ is the subset of vertices in $\upmrkd(i)$ that are raised to level $k > i$. We then show that the expected number of tokens released as result of raising each $R_k$ is proportional to the work done by raised vertices, i.e. $\sum_k 3^k \E[|R_k|]$. Roughly, the raising procedure works by processing $\upmrkd(i)$ in rounds wherein each round a subset $M$ of $\upmrkd(i)$ is raised to a level $k > i$. For $u \in M$, let $N_u^0$ and $N_u^1$ denote the neighborhoods of $u$ before and after raising $M$ respectively. The algorithm is such that each vertex $u$ in $M$ satisfies 
\begin{enumerate}[1.]
    \item $|N_u^0(1,k-1)| \geq 3^{k-1}$;
    \item $|N_u^0(1,k)| \leq 3^k$; 
    \item $|N_u^0[M]| \leq \alpha$
\end{enumerate}
where $\alpha := 3^{k-1}/4$. Together, these properties imply that $|N_u^1(1,k-1)| \geq 3^{k-1} - \alpha$ and $|N_u^1(1,k)| \leq 3^k$, which is exactly what we shall require to show that moving $u$ from $i$ to $k$ releases an appropriate number of tokens.

Now we describe the procedure in more detail. We iterate from level $k = \lambda$ down to level $i+1$. Prior to this iteration, we copy $\upmrkd(i)$ to an array denoted by $U$ and, for each vertex in $U$, compute the prefix sum array $S_u$ where $S_u[k] = |N_u(1,i-1)| +  \sum_{j=i}^{k-1} |N_u(j)|$. We additionally record the initial value of $|N_u(i)|$ which we denote by $n_u(i)$. We will maintain a partition of $U$ into $U^-$ and $U^+$ based on the predicate
\[P(u,k) = \Ind{|N_u(1,k-1)| \geq 3^{k-1} \land |N_u(1,k)| \leq 3^k}.\]
Note we can compute $|N_u(1,k-1)|$ and $|N_u(1,k)|$ using $O(1)$ work since $|N_u(i)|$ and $|N_u(k)|$ are maintained dynamically (from hash table representations) and 
\[ |N_u(1,k-1)| = S_u[k] - (n_u - |N_u(i)|) \quad \text{and} \quad |N_u(1,k)| = |N_u(1,k-1)| + |N_u(k)|.
\]
Next, we describe how to handle a single iteration of processing level $k$. We focus on the algorithmic details and defer the technical details of the data structure updates to Section \ref{sec:raising-procedure-updates}. We first filter out each vertex $u$ in $U^+$ such that $P(u,k) = 0$ by partitioning $U^+$ and adding the zeros to $U^-$. At this point, each vertex in $U^+$ satisfies $|N_u(1,k-1)|\geq 3^{k-1}$ and $|N_u(1,k)| \leq 3^k$. Now we select each vertex of $U^+$ uniformly at random with probability $1/24$. In parallel, we move each vertex that was both selected and had at most $\alpha$ of its neighbors in $U^+$ selected. We denote this set of moved vertices by $M$. Observe by construction that $M$ satisfies the desired properties listed above. Finally, we partition out the moved vertices from $U^+$ and update the relevant data structures. We repeat this procedure until $U^+$ is empty at which point we set $U^+ = U^-$ and decrement $k$.

\begin{algorithm}[ht]
\function{$\emph{\textsc{Raise}}(\upmrkd(i))$}{
    $U \gets \textsc{List}(\upmrkd(i))$.\\
    Set $U^- \gets \{\}$ and $U^+ \gets U$.\\
    \For{$k = \lambda$ to $i + 1$}{
        \While{$U^+$ is not empty}{
            $z \gets \textsc{Partition}(U^+, P(u,k))$.\Comment{returns \# of zeros.} \\
            $U^- \gets U[0:|U^-| + z - 1]$ and $U^+ \gets U^+[z:|U|-1]$.\\ 
            $S \gets \textsc{Sample}(U^+, 1/24)$. \\
            $M \gets \textsc{Filter}(S, P(u))$ where $P(u) = \Ind{|N_u[S]| \leq \alpha}$. \\
            Move $M$ to level $k$ and update the relevant data structures. * \\
            $z \gets \textsc{Partition}(U^+, \Ind{u \in M})$.\\
            $U^+ \gets U^+[0:z-1].$
        }
        $U^+ \gets U^-$ and $U^- \gets \{\}$.
    }
}
\caption{Raising Procedure. * See Section \ref{sec:raising-procedure-updates}.}\label{alg:Raise}
\end{algorithm}

\subsubsection{Raising procedure analysis.} Recall that for $u \in M$, $N_u^0$ denotes the neighborhood of $u$ the moment before $M$ is raised and $N_u^1$ denotes the neighborhood the moment after. The following Lemma asserts that $M$ satisfies the three properties stated at the start of Section \ref{sec:raising-procedure}.
\begin{lemma}\label{lem:raising-correctness}
    For all $u \in \upmrkd(i)$, if after $\emph{\textsc{Raise}}(\upmrkd(i))$, $\ell(u) = i$, then
    \[|N_u(1,i)| \leq 3^{i}.\]
    Otherwise, if $u$ was raised with the set $M$, then
    \[|N_u^0(1,k-1)| \geq 3^{k-1}, \quad |N_u^0(1,k)| \leq 3^k, \quad \text{ and } \quad |N_u^0[M]| \leq \alpha.\]
\end{lemma}
\begin{proof}
    The latter statement is true by virtue of how the moved set is selected. We show by a simple induction that $|N_u(1,j)| \leq 3^j$ for $i \leq j \leq \lambda$ after processing level $j$, from which our claim follows. The base case of level $\lambda$ is immediate since $|N_u(1,\lambda)| \leq 3^{\lambda} = \Delta$. Now, suppose that $|N_u(1,j)| \leq 3^{j}$ for $j > i$. In processing level $j$, since $u$ is not moved, this means $P(u,j) = 0$ at some point, which then implies that $|N_u(1,j-1)| \leq 3^{j-1}$. We conclude that $|N_u(1,j)| \leq 3^j$ for all $i \leq j \leq \lambda$.
\end{proof}

Next, we analyze the complexity of the Algorithm \ref{alg:Raise}. The following Lemma bounds the work and span of Algorithm \ref{alg:Raise}.
\begin{lemma}\label{lem:raise-work}
    The raising procedure described in Algorithm \ref{alg:Raise} can be implemented using 
    \[O\left(\lambda |\upmrkd(i)| + \sum_{k=i+1}^{\lambda} 3^{k}\E[|R_k|]\right)\]
    expected work and $O(\log^3 n)$ span \whp.
\end{lemma}
\begin{proof} First we bound the work of Algorithm \ref{alg:Raise}. We start by introducing some notation. We define, for $i \leq k \leq m$, $R_k \subseteq U$ to be the vertices which are raised to level $k$. For each $t \geq 1$, we define $U_{k,t}^+$ to be $U^+$ at the start of iteration $t$ of processing level $k$ and $U_{k,t}^{++} \subseteq U_{k,t}^+$ to be the vertices of $U_{k,t}^+$ such that $P(u,k) = 1$. Additionally, we define $R_{k,t}$ and $S_{k,t}$ to be $M$ and $S$ respectively during iteration $t$ of processing level $k$. Finally, we denote by $\operatorname{DS}_{k,t}$ the collection of update requests made during iteration $t$ of processing level $k$. 

The work to initialize $U$, denote $W_\text{init}$ is simply the size of $\upmrkd(i)$. Next, we analyze the excepted work of the for loop, say $W_\text{main}$. We can write this as
\begin{equation}\label{eqn:expected-main-work-raising}\E[W_\text{main}] = \sum_{k=i+1}^\lambda \sum_{t=1}^\infty\E[W_{k,t}]\end{equation}
where $W_{k,t}$ denotes the work performed during iteration $t$ of processing level $k$. During iteration $t$, we do $O(\E[U^+_{k,t}])$ work to partition $U^+$ (line 6) and $O(\E[U^{++}_{k,t}])$ work during the call to sample (line 8) and partition (line 11). Additionally, when computing $M$ (line 9), we do $O(\E[|N_u(i)|])$ work for each $u \in S_{k,t}$. Combined with the work for data structure updates, we can express the expected work for $W_{k,t}$ as
\begin{multline}\label{eqn:raising-procedure-main-work-bound}\E[W_{k,t}] = \E[W(\operatorname{DS}_{k,t})] + O(\E[|U_{k,t}^+|]) + \sum_{u \in S_{k,t}}\E[|N_u(i)|] \\\leq E[W(\operatorname{DS}_{k,t})] + O(\E[|U_{k,t}^+|]) + 3^k \E[|U_{k,t}^{++}|]
\end{multline}
where $W(\operatorname{DS}_{k,t})$ denotes the work required to perform the updates in $\operatorname{DS}_{k,t}$. Note the inequality follows from the fact that each $u \in S_{k,t}$ had $P(u,k) = 1$ and $S_{k,t} \subset U_{k,t}^{++}$. Next, we bound the work to carry out $\operatorname{DS}_{k,t}$ by Lemma \ref{lem:rasining-ds-work-1} in Section \ref{sec:raising-procedure-updates} which gives
\begin{equation}\label{eqn:raising-procedure-data-structure-bound-1}\E[W(DS_{k,t})] = O(3^k \E[|R_{k,t}|]) \end{equation}
To conclude the analysis of $W_\text{main}$, we observe the following bounds
\begin{equation}\label{eqn:raising-work-proof-bounds-1}\E[|U_{k,t}^{++}|] = c_1 \E[|R_{k,t}|] \quad \text{ and } \quad \E[|U_{k,t}^+|] \leq c_2^t \E[|U_{k,1}^+|] \end{equation}
for absolute constants $c_1,c_2$ where $c_2 < 1$. Assuming the bounds in \ref{eqn:raising-work-proof-bounds-1} and using \ref{eqn:raising-procedure-main-work-bound} and \ref{eqn:raising-procedure-data-structure-bound-1}, we have
\begin{align*}
\E[W_{\text{main}}]
&\le \sum_{k=i+1}^\lambda \sum_{t=1}^{\infty}
   \Bigl(O(3^k \E[|R_{k,t}|]) + O(c_2^t\E[|U_{k,1}^+|])\Bigr) \\
&\le O\left(\sum_{k=i+1}^\lambda \E[U_{k,1}^+] + \left(\sum_{t=1}^{\infty}3^k\E[|R_{k,t}|]\right)\right)\\
&\le O\left(\lambda|\upmrkd(i)|
   + \sum_{k=i+1}^{\lambda} 3^{k} \E[|R_k|]\right)
\end{align*}
where the last inequality follows from the trivial bound of
\[\sum_{k=i+1}^{\lambda} \E[|U_{k,1}^+|] \leq \lambda |\upmrkd(i)|.\]
Putting everything together we can bound the total expected work $W$ as
\[\E[W] = \E[W_\text{init}] + \E[W_\text{main}] = O\left(\lambda |\upmrkd(i)|+ \sum_{k=i+1}^{\lambda} 3^{k}\E[|R_k|]\right).\]
Now let us prove the bounds stated in \ref{eqn:raising-work-proof-bounds-1}. Recall that each vertex in $U_{k,t}^{++}$ is included in $R_{k,t}$ only if it is sampled (in $S_{k,t}$) and at most $\alpha$ of its neighbors in $U_{k,t}^{++}$ are sampled. Additionally, recall that $\alpha = 3^{k-1}/4$ and and each vertex is sampled with probability $1/24$. For $u \in U_{k,t}^{++}$, let $X_u$ be the indicator for $u$ being sampled and $Y_u$ be the sum of the indicators over $u$'s neighbors, that is 
\[Y_u := \sum_{v \in N_u[U_{k,t}^{++}]} X_v.\]
Then
\[\Pr[u \in R_{k,t}] = \Pr[X_u = 1]\Pr[Y_u \leq \alpha] = \Pr[Y_u \leq \alpha]/24.\]
Since $\E[Y_u] = |N_u[U_{k,t}^{++}|]/24 \leq 3^{k}/24$, we have by Markov's inequality that $\Pr[Y_u \geq \alpha] \leq 3^k/24\alpha = 1/2$ and thus $\Pr[u \in R_{k,t}] \geq 1/48 := 1/c_1$. Therefore, conditioned on the size of $U_{k,t}^{++}$, the expected size of $R_{k,t}$ is $|U_{k,t}^{++}|/c_1$. This proves the first bound in \ref{eqn:raising-work-proof-bounds-1}. To prove the second bound, observe that as random sets $U_{k,t}^+ = U_{k,t-1}^{++} \setminus R_{k,t}$. Then, by the above analysis for the expected size of $R_{k,t}$, for $c_2 = 47/48$
\[\E[|U_{k,t}^+|] = c_2 \E[|U_{k,t-1}^{++}] \leq c_2 \E[U_{k,t-1}^+].\]
This proves the second bound. Next, we analyze the span of Algorithm \ref{alg:Raise}. Each iteration of the inner while loop performs two partitions, a sample, and a filter each of which requires $O(\log n)$ span. Additionally, by Lemma \ref{lem:rasining-ds-work-1}, the span to perform $\operatorname{DS}_{k,t}$ is $O(\log n)$ \whp. Thus, each iteration of the while loop uses $O(\log n)$ span \whp. Next, let us analyze the number iterations of the while loop for level $k$, i.e. the number iterations until $U_{k,t}^+$ is empty. By the second bound in \ref{eqn:raising-work-proof-bounds-1} and Markov's we have that
\[\Pr[|U_{k,t}^+| \geq 1] \leq \E[|U_{k,t}^+|] \leq c_2^t \E[|U_{k,1}^+|]\]
thus for any $c_3 > 1$ and $t > (c_3+1)\log_{1/c_2} n$, we have that
\[\Pr[|U_{k,t}^+| \geq 1] \leq \frac{\E[|U_{k,1}^+|]}{n^{c_3+1}} \leq \frac{1}{n^{c_3}}\]
since $\E[|U_{k,1}|] \leq n$. Therefore, \whp, \ after $O(\log n)$ iterations, $U_{k,t}$ will be empty. Finally, since there at most $\lambda - i = O(\log n)$ levels processed, and each level uses $O(\log^2 n)$ span \whp, we conclude the lemma.
\end{proof}

\subsubsection{Raising procedure amortized analysis.} Recall that $R_k$ is the subset of vertices from $U$ which are raised to level $k$ and $R_{k,t} \subseteq R_k$ are the vertices which move to level $k$ in round $t$ of processing level $k$. We show that each movement of $R_{k,t}$ releases $\Theta(3^{k} |R_k^t|)$ tokens.

First, we determine the token change contributed by the vertex token functions. Fix $u \in R_{k,t}$. Observe that moving $u$ from level $i$ to level $k$ only affects the values of the vertex token functions for vertices in $N_u(1,k)$. Also, moving $u$ from $i$ to $k$ increases $\theta(u)$ by at most $|N_u[R_{k,t}]|/6 \leq \alpha/6$ as each $u \in R_{k,t}$ satisfies $|N_u[R_{k,t}]| \leq \alpha$. Next, for $x \in N_u(i+1, k)$, moving $u$ above or to the same level as $x$ increases $\theta(x)$ by at most $1/6$. Thus in total, the vertex token increase is at most 
\[\sum_{u \in R_{k,t}} \frac{|N_u[R_{k,t}]|+ |N_u(i+1,k)|}{6} \leq |R_{k,t}|\frac{\left(\alpha + 3^{k}\right)}{6}\]
where we use $|N_u(i+1,k)| \leq 3^k$ and $|N_u[R_{k,t}]| \leq \alpha$ by definition of $R_{k,t}$. Now, we aim to understand the token change from the edges. Again fix $u \in R_{k,t}$. Similarly to the vertex tokens, moving $u$ only affects the tokens of the edges the form $(u,x)$, with $x \in N_u(1,k-1)$. For each $v \in N_u[R_{k,t}]$, the token potential $\theta(u,v)$ decreases by $k-i \geq 1$, and for each $x \in N_u(j)$ with $i < j < k-1$,  the token potential $\theta(u,x)$ decreases by $k - j \geq 1$. Summing over all relevant edges, we can lower bound the token decrease from edges by
\[\left(\sum_{u \in R_{k,t}} |N_u(1,k-1)|\right) - |E[R_k^t]| \geq |R_{k,t}|(3^{k-1} - \alpha/2)\]
where again the inequality follows by definition of $R_{k,t}$. Combining both vertex and edge token changes the total token decrease from moving $R_{k,t}$ is at least 
\[
    |R_{k,t}|(3^{k-1} - \alpha/2 - 1/6(\alpha +  3^k))  = |R_{k,t}|(3^{k-1}/2 - 2\alpha/3).
\]
Then, since $\alpha \leq 3^{k-1}/4$, we can say that moving $R_{k,t}$ releases at least $3^{k-2}|R_{k,t}|$ tokens. Then summing over $t$, the total token decrease generated by Algorithm \ref{alg:Raise} in processing level $k$ is at least
\[
    3^{k-2}|R_k| = \Theta(3^k|R_k|).
\]
Summing over $k$, the total decrease for raising $\upmrkd(i)$ comes out to $\Theta(\sum_{k=i}^\lambda 3^k |R_k|)$. Since this a sum of random variables, we can write the expectation as
\[\Theta\left(\sum_{k=i+1}^\lambda 3^k \E[|R_k|]\right).\]
We conclude the following Lemma.

\begin{lemma}\label{lem:adam-1}
    Calling $\textsc{Raise}(\upmrkd(i))$ releases
    \[\Theta\left(\sum_{k=i+1}^\lambda 3^k \E[|R_k|]\right)\]
    tokens in expectation, where $R_k$ is the subset of vertices from $\upmrkd(i)$ that are raised to level $k$.
\end{lemma}


\subsubsection{Raising procedure updates.}\label{sec:raising-procedure-updates} We detail the updates made for a single iteration of processing level $k$. 

We describe the updates made in raising $M$ (line 10 in Algorithm \ref{alg:Raise}). Fix $u \in M$ and let $c := \chi(u)$. For each vertex $v$ in the set $N_u(1,k)$, we will make update requests to account for $u$'s upward movement from $i$ to $k$. First, if $v \in N_u(1,i)$, then we make a request to remove $u$ from $\nbr{v}{i}$ and a request to add $u$ to $\nbr{v}{k}$. Note after this movement, $u$ remains above or at the same level as $v$ so we do not need to update color palettes. If $v \in N_u(j)$ for some $i < j \leq k$, then $u$ was previously below $v$. Thus we make requests to move up $c$ in $\mc C_v$, increment $v$'s color counter for $c$, delete $u$ $\belowHT{v}$, and add $u$ to $\nbr{v}{k}$. At this point, we apply the above update requests which concludes the non-local updates required for $u$.

Next, we make updates which are local to $u$. We set $u$'s level $\ell(u)$ to be $k$. Next, for each $v \in N_u(1,k-1) \setminus M$, we make an update request to decrement $\mu_u^+(\chi(v))$. If after applying these decrement requests, $\mu_u^+(\chi(v))$ drops to $0$ then we make an additional update request to move down $\chi(v)$ in $\mc C_u$.  Finally, after applying the move request, we create $u$'s new lower neighborhood. This is accomplished by sequentially iterating from $t = i$ to $k-1$ and batch inserting $\nbr{u}{t}$ into $\belowHT{u}$ and then deleting it. 

Recall the notation from the proof of Lemma \ref{lem:raise-work}. The following Lemma summarizes the complexity of performing the required data structure updates during iteration $t$ of processing level $k$.

\begin{lemma}\label{lem:rasining-ds-work-1} We have
    \[\E[W(\operatorname{DS}_{k,t})] =  O(3^k \E[|R_{k,t}|]) \quad \text{ and } \quad D(\operatorname{DS}_{k,t}) = O(\log n) \; \emph{whp}.\]
\end{lemma}

\subsubsection{Lowering Procedure Description.}\label{sec:lower-procedure} We describe a procedure to lower vertices in $\lmrkd(i)$. The pseudocode for this procedure is given in Algorithm \ref{alg:Lower}. In what follows, we describe Algorithm \ref{alg:Lower} and show that the expected work is bounded by
\begin{equation}\label{eqn:lower-work-statement}
    O(3^i|\lmrkd(i)|) 
\end{equation}
We then show that the expected number of tokens released as a result of the lowering procedure is proportional to $\ref{eqn:lower-work-statement}$. 

Structurally, the lowering procedure is simpler than the raising procedure. Let us fix a subset $L$ of $\lmrkd(i)$ and recall that $N_u^0$ and $N_u^1$ denote the neighborhoods of $u$ before and after moving $L$ respectively. For our token argument, we want to lower each vertex in $L$ to a level $k \leq i-4$ such that for $|N^1_u(1,f)| \leq 3^f$ for each $k \leq  f \leq i-1$. Importantly, we are not concerned with whether $u$ satisfies the lower condition \ref{con:lower} at such a level $k$. Thus, if we ensure that $|N_u[L]| \leq 3^{i-4}$, then since $|N^0_u(1,i-1)| < 3^{i-5}$ (by definition of being in $\lmrkd(i)$), moving $L$ to level $i-4$ will be sufficient to guarantee what we want. What remains is to describe how to compute a large $L$. We do this using the same symmetry breaking condition from the raising procedure: sample each vertex in $\lmrkd(i)$ with probability $1/(2\cdot 3^6)$. Then each vertex that is sampled and has at most $3^{i-5}$ of its neighbors sampled is moved to level $i-4$, i.e. is included in $L$.

\begin{algorithm}[ht]
\function{$\emph{\textsc{Lower}}(\lmrkd(i))$}{
        $S\gets \textsc{Sample}(\lmrkd(i), 1/(2\cdot 3^6))$. \\
        $L \gets \textsc{Filter}(S,P(u))$ where $P(u) = \Ind{|N_u[S]| \leq 3^{i-5}}$. \\
        Move $L$ to level $i-4$. \\
        Update the relevant data structures$.^*$
    
}
\caption{Lowering Procedure. \hfill ${ }^*$See Section \ref{sec:lowering-procedure-ds-updates}.}\label{alg:Lower}
\end{algorithm}


\subsubsection{Lowering procedure analysis.} 


Recall that $L$ is the subset of $\lmrkd(i)$ which is lowered to level $i-4$. The following Lemma bounds the work and span of Algorithm \ref{alg:Lower}.

\begin{lemma}\label{lem:lower-work}
    The lowering procedure described in Algorithm \ref{alg:Lower} can be implemented in
    \[
      O\left(3^i|\lmrkd(i)|\right)
    \]
    expected work and $O(\poly n)$ span \emph{whp}.
\end{lemma}
\begin{proof}
    The work and span for Algorithm \ref{alg:Lower} can be split into computing $S$ (line 2) and $L$ (line 3) and performing data structure updates. By Lemma \ref{lem:lowering-ds-updates-complexity}, the work and span to perform the required data structure updates is $O(3^i |\lmrkd(i)|)$ and $O(\log n)$ whp. Next, we analyze the work and span to compute $S$ and $L$. The work and span for $\textsc{sample}(\lmrkd(i),1/(2 \cdot 3^6))$ is $O(\lmrkd(i))$ and $O(\log n)$ respectively. Finally, $\textsc{filter}(S, P(u))$ requires work and span proportional to $O(3^i|S|)$ work and $O(\log |S| + \log n) = O(\log n)$ span as evaluating the indicator requires $O(|N_u(i)|) = O(3^i)$ work and $O(\log n)$ span. Note we used that $|N_u(i)| \leq 3^{i+2}$ owing to the fact that each $u \in \lmrkd(i)$ satisfies the upper condition \ref{con:upper}. 
    
\end{proof}

\subsubsection{Lowering procedure amortized analysis.} Recall that $L$ is the subset of $\lmrkd(i)$ which is lowered to level $i-4$. We first show that $\E[|L|] = \Theta(|\lmrkd(i)|)$. Recall that each vertex in $\lmrkd(i)$ is included in $L$ only if it is sampled and at most $3^{i-4}$ of its neighbors in $\lmrkd(i)$ are sampled. For $u \in \lmrkd(i)$, let $X_u$ be the indicator for $u$ being sampled and $Y_u$ be the sum of the indicators over $u$'s neighbors in $\lmrkd(i)$. Then
\[\Pr[u \in L] = \Pr[X_u = 1]\Pr[Y_u \leq 3^{i-5}] = (2\cdot 3^6)^{-1}\Pr[Y_u \leq 3^{i-5}].\]
Since $\E[Y_u] = (2\cdot 3^6)^{-1}|N_u[\lmrkd(i)]| \leq 3^{i-4}/2$, we have by Markov's inequality that $\Pr[Y_u \geq 3^{i-4}] \leq 1/2$ and thus $\Pr[u \in L] \geq (4\cdot 3^6)^{-1}$.

Next, we show that the number of tokens released by the lowering of $L$ is $\Theta(3^i|L|)$ tokens. First, we determine the token change contributed by the vertex token functions. Fix $u \in L$.
Lowering vertices does not increase any vertex token function by definition \ref{eqn:vertex-potential}. Moreover, initially $\theta(u)$ was $(3^{i-1} - |N_u(1,i-1)|)/6 \geq (3^{i-1} - 3^{i-5})/6$.
After, the movement, $u$ is at level $i-4$, and thus $\theta(u) \leq 3^{i-5}/6$. Therefore, the vertex token decrease is at least 
\[|L|\left(\frac{3^{i-2}}{2}\right).\]
Now, we analyze the token change from edges. Again fix $u \in L$. From \ref{eqn:edge-potential}, the number of tokens associated with the edge $(u,v)$ for each $v \in N_u(1,i-1) \setminus L$, increases by $i - \max(\ell(v), i-4) \leq 4$. Additionally, for each $v \in N_u[L]$, $\theta(u,v)$ increases by $4$. Therefore, we can upper bound the total increase in edge tokens by
\[
\left(\sum_{u \in L} 4 |N_u(1,j) \setminus L|\right) + \frac{1}{2}\sum_{u \in L} 4|N_u[L]| < 4\cdot 3^{i-5}|L| + 2\cdot3^{i-4}|L| = 10/3\cdot3^{i-4}|L|
\]
where the middle inequality follows from the fact that for $j \leq i-1$, $|N_u(1,j)| \leq |N_u(1,i-1)| < 3^{i-5}$ for each $u \in \lmrkd(i)$ prior to lowering $L$ and $|N_u[L]| \leq 3^{i-4}$ by construction. Therefore, the total decrease is at least
\[|L|\left(\frac{3^{i-1} - 20 \cdot 3^{i-4}}{6}\right) = |L|\left(\frac{7 \cdot 3^{i-4}}{6}\right) = \Theta(3^i|L|).\]
Finally, since $\E[|L|] = (4\cdot 3^6)^{-1} |\lmrkd(i)|$, we have that the expected number of tokens released is
$\Theta(3^i\E[|\lmrkd(i)|])$. We restate this the following Lemma.

\begin{lemma}\label{lem:adam-2}
    Call $\textsc{Lower}(\lmrkd(i))$ releases 
    \[\Theta(3^i\E[|\lmrkd(i)|])\]
    tokens in expectation.
\end{lemma}





\subsubsection{\textbf{Lowering procedure updates.}}\label{sec:lowering-procedure-ds-updates} Recall that $L$ is the subset of $\lmrkd(i)$ which is lowered to level $i-4$. We describe the data structures updates necessitated by $L$'s movement. 

Fix $u \in L$ and let $c := \chi(u)$. For each $v$ in the set $ N_u(1,i)$, we will create update requests to account for $u$'s downward movement from $i$ to $i-4$. If $v \in N_u(1,i-1)$ and $\ell(v) \leq i-4$, then, we will make a request to remove $u$ from $\nbr{v}{i}$ and a request to add $u$ to $\nbr{v}{i-4}$. Note after this movement, $u$ remains above or at the same level as $v$ so we do not need to update color data. On the other hand, if $v \in N_u(j)$ for some $i-4 < j \leq i$ and $v \not\in L$, then $u$ was previously above $v$ and ends up below it. Thus we make requests to remove $u$ from $\nbr{v}{i}$, add $u$ to $\belowHT{v}$, and decrement $\mu_v^+(c)$. Finally, if $v \in N_u(i) \cap L$, that is $v$ is a level $i$ neighbor of $u$ which is also moved down to level $i-4$, then we make a request to remove $u$ from $\nbr{v}{i}$ and add $u$ to $\nbr{v}{i-4}$. At this point we apply the above update requests. If after applying these requests, $\mu_v^+(c) = 0$, then we make an additional update request to move down $c$ in $\mc C_v$. Applying these requests concludes the non-local updates made by $u$. 

Next, we describe updates which are local to $u$. We set $u$'s level $\ell(u)$ to be $i-4$. For each $v \in N_u(i-4,i-1) \setminus L$ we make a request to increment $\mu_u^+(\chi(v))$ and a request to move up $\chi(v)$ in $\mc C_u$. Finally, after applying these request, we create $u$'s new lower neighborhood list. To do this, we compute $N_u(i-4, i-1)$ and batch delete it from $\belowHT{u}$.

The following Lemma summarizes the complexity of performing the required data structure updates.

\begin{lemma}\label{lem:lowering-ds-updates-complexity}
     The work and span to perform the required data structure updates is $O(3^i |\lmrkd(i)|)$ and $O(\log n)$ \whp.
\end{lemma}

\subsection{The complete algorithm}\label{sec:full-algorithm}
Initially, when the graph $G = (V,E)$ is empty, every vertex $u$ belongs to level 5 and is colored with a random color from $\mc C = \{0,\dots,\Delta\}$. At this point, the coloring $\chi$ is proper as there are no edges. Let $S$ be a batch of updates and suppose that prior to applying this batch $\chi$ is a proper $\Delta + 1$ coloring of $G$.

Our update algorithm is as follows. To start, we apply $S$ to $G$ and compute the marked and unmarked level sets $\upmrkd(i),\lmrkd(i),$ and $\umrkd(i)$ according to the procedure in Section \ref{sec:init-phase}. Next, we color every blank vertex by iterating top-down over the levels of the partition using Algorithm \ref{alg:mrkd-coloring} and \ref{alg:umrkd-coloring} to color $\mrkd(i)$ and $\umrkd(i)$ respectively. At this point, the coloring $\chi$ is proper. To pay for the coloring we perform two additional passes over the levels of the partition. In the first pass, we iterate top-down and process $\upmrkd(i)$ using Algorithm \ref{alg:Raise} for each level $i$. Then, in the second pass, we iterate bottom-up and process $\lmrkd(i)$ using Algorithm \ref{alg:Lower} for each level $i$. The pseudocode for this procedure is given in Algorithm \ref{alg:final-alg}.

\begin{algorithm}[ht]
    \function{$\emph{\textsc{Update}}(S)$}{
        Apply $S$ and initialize $\upmrkd(i),\lmrkd(i)$, and $\umrkd(i)$ for each $i \in [5,\lambda]$ as described in Section \ref{sec:init-phase}. \\
        \For{$i = \lambda$ to $5$}{
            $\textsc{ColorMarked}(\mrkd(i))$. \\
            $\textsc{ColorUnmarked}(\umrkd(i))$.\\
        }

        \For{$i = \lambda$ to $5$}{
            $\textsc{Raise}(\upmrkd(i))$. \\
        }

        \For{$i = 5$ to $\lambda$}{
            $\textsc{Lower}(\lmrkd(i))$. 
        }
    }

\caption{Update Algorithm}\label{alg:final-alg}
\end{algorithm}

\subsection{Full analysis}\label{sec:full-analysis}

The goal of this section is to establish the main theorem using our algorithm.
\begin{theorem}\label{thm:main-thm}
There exists a randomized algorithm that maintains a proper $(\Delta+1)$-vertex coloring in a dynamic $\Delta$-bounded graph using $O(\log \Delta)$ expected amortized work per update and, for any batch of $b$ updates, has parallel span $O(\poly b + \poly n)$ with high probability.
\end{theorem}
Fix a sequence of batch updates $S_1, \dots, S_T$. First we analyze the span of $\textsc{Update}(S_i)$ for each $i \leq T$. By Lemma \ref{lem:init-work-span} the initialization phase uses $O(\log n)$ span \whp. Next, for each level $i$, we have by Lemma \ref{lem:raise-work}, \ref{lem:lower-work}, and \ref{lem:coloring-work-span} that the total span for the $\textsc{ColorMarked}(\mrkd(i))$, $\textsc{ColorUnmarked}(\umrkd(i))$,  $\textsc{Raise}(\upmrkd(i))$, and $\textsc{Lower}(\lmrkd(i))$ procedures is $O(\log^3 n)$ \whp. Thus in total, over all the levels, the span for $\textsc{Update}(S_i)$ is $O(\log^3 n \log \Delta)$ \whp.

Next, to understand the total amortized-expected work of the algorithm over the sequence $S_1,\dots,S_T$, we setup the argument in the same way that we did for the relaxed sequential algorithm. To that end, we first recall the notion of a \emph{record}. During the execution of the algorithm, every time a vertex $u$ is recolored, we record the timestamp in a sequence. We call this sequence the \emph{record} of $u$, and denote it by $R_u = (\tau_u^0, \tau_u^1, \dots, \tau_u^k)$. We refer to the entire record as $R$.

We now make the following important definition that will feature throughout the entire analysis. The \emph{base level} of $u$ at time $\tau_u^i$ is the quantity
\[
    b_i(u) = \max(\lfloor \log |N_u(1,\ell_i(u))| \rfloor,\ell_i(u)),
\]
where $\ell_i(u)$ is the level of $u$ at time $\tau_u^i$. For an arbitrary timestamp $\tau$, we denote the base level by $b(\tau)$. We record the following important property of the base level.
\begin{lemma}\label{lem:base-level-prop}
    The vertex $u$ violates the upper condition at time $\tau_u^i$ if and only if $b_i(u) \geq \ell_i(u)+2$.
\end{lemma}

We associate the following cost to each timestamp $\tau_u^i$:
\[
    w(\tau_u^i) := 3^{\max\{b_i(u),\ell_{i+1}(u)\}},
\]
where $\ell_{i+1}(u)$ is the level $u$ ends up on after potentially moving at time $\tau_u^i$. We also make the following convention.  For any subset $X$ of timestamps, let $X(k)$ denote the timestamps $\tau_u^i \in X$ such that $\ell_i(u) = k$.

Our next goal is to relate the total work over all batches of updates to the cost of the record. We first prove the following Lemma.
\begin{lemma}\label{lem:marked-stable}
For each level $i$:
\begin{enumerate}
    \item During the raising pass, the set $\upmrkd(i)$, and their base levels, are unchanged until level $i$ is processed.
    \item During the lowering pass, the set $\lmrkd(i)$ is unchanged until level $i$ is processed.
\end{enumerate}
\end{lemma}

\begin{proof}
In the raising pass, vertices only move \emph{up}, from level $i$ to some level $k>i$, and levels are processed in decreasing order. Thus, before level $i$ is processed, no vertex currently on level $i$ has moved, and no vertex from a lower level can move into level $i$. Moreover, raising vertices on levels above $i$ never decreases $|N_u(1,i)|$ for any vertex $u$ on level $i$. Hence $\upmrkd(i)$ does not change until processed.

In the lowering pass, vertices only move \emph{down}, and levels are processed in increasing order. Before level $i$ is processed, no vertex on level $i$ has moved, and no vertex from a higher level can move into level $i$. Lowering vertices on levels $<i$ never decreases $|N_u(1,i-1)|$ for any vertex $u$ on level $i$, so $\lmrkd(i)$ is also unchanged until level $i$ is processed.
\end{proof}

With the above Lemma in hand, we now relate the work of the algorithm to the work of the record.
\begin{lemma}\label{lem:total-work-parallel}
    Let $W_T$ denote the total work of the algorithm over the first $T$ batches of updates $S_1, \dots, S_T$. Then
    \[
        \E[W_T] = O(\lambda T) + O\left( \sum_u \sum_{\tau \in R_u(T)} w(\tau) \right) = O(\lambda T) + O(w(R)).
    \]
\end{lemma}
\begin{proof}
    The result immediately follows from the definition of $w(\tau_u^i)$, in combination with Lemma \ref{lem:marked-stable}, and Lemmas \ref{lem:init-work-span}, \ref{lem:coloring-work-span}, \ref{lem:raise-work}, and \ref{lem:lower-work}.
\end{proof}
With this lemma established, we proceed by understanding the cost of the record $w(R_T)$ in the same way as in the sequential case. As in that case, we determine that the cost of the entire record is proportional to the cost of the \emph{important} timestamps $I$, where an \emph{important} timestamp is one which corresponds to a clean spawning or dirty terminating vertex of a chain. According to this distinction, we partition $I$ into clean and dirty sets denoted $C_T$ and $D_T$. We may further partition $C_T$ into two sets $CC_T$ and $DC_T$, where the $CC_T$ corresponds to timestamps $\tau_u^i \in C_T$ such that $\tau_u^{i-1} \in C_T$, and  $DC_T$ corresponds to $\tau_u^i \in C_T$ such that $\tau_u^{i-1} \in D_T$. Also observe that $|I|$ is at most $2T$, since each chain gives at most $2$ important timestamps.

We now relate the cost of $DC_T$ to that of $D_T$.
\begin{lemma}
    We have
    \[
        w(DC_T) \leq w(D_T).
    \]
\end{lemma}
\begin{proof}
    Each $\tau_{u}^{i} \in DC_T$ is clean with a dirty predecessor $\tau_u^{i-1}$. Thus \[
    w(\tau_u^i) = 3^{\ell_{i}(u)} \leq 3^{\max\{b_{i-1}(u),\ell_i(u)\}} =w(\tau_u^{i-1}). 
    \]
    Taking predecessors uniquely associates a member $\tau^-$ of $D_T$ to each member $\tau$ of $DC_T$, so we get
    \[
        w(DC_T) = \sum_{\tau \in DC_T} w(\tau) \leq \sum_{\tau \in DC_T} w(\tau^-) \leq \sum_{\tau' \in D_T} = w(D_T),
    \]
    as desired.
\end{proof}
We thus have
\[
    w(R_T) = w(CC_T) + 2w(D_T).
\]
The analysis now splits into understanding the first term, the \textbf{clean part}, and the second term, the \textbf{dirty part}.

\paragraph{Clean part analysis.} We use a variant of the deferred decision argument. To do this, we interpret the cost of the clean part from the perspective of edge updates. We define the timestamp $\tau_e$ to be the timestamp that would result if the edge $e$ would be monochromatic upon insertion. We also let $\tau_e^-$ denote the predecessor of $\tau_e$ in its appropriate record. We now aim to prove the following lemma.
\begin{lemma}\label{lem:def-dec}
    We have
    \[\E[w(CC_T)] = O(T).\]
\end{lemma}
\begin{proof}
First, we have the equality
\begin{equation}\label{eq:cc-work-identity}
    w(CC_T) = \sum_{e \in S_T} \ind[\tau_e \in CC_T] \cdot w(\tau_e).
\end{equation}
Let $x_e$ denote the endpoint of $e$ that would be most recently recolored upon insertion, and $\ell(e)$ the level that $x_e$ would be at upon insertion. Also let $\ell^-(e)$ be the level of $x_e$ at time $\tau_e^-$. Finally, let $\ind[\mathrm{mono}]$ denote the indicator for the event that $e$ is monochromatic upon insertion, let $\ind[C_e]$ be the indicator for the event that $\tau_e \in C_T$, and $\ind[C_e^-]$ be the indicator for the event that $\tau_e^- \in C_T$. All three of these events are necessary for $\tau_e$ to be in $CC_T$, so \[
    \Ind{\tau_e \in CC_T} \cdot w(\tau_e) \leq \ind[\mathrm{C_e^-}] \cdot \ind[\mathrm{mono}] \cdot \ind[C_e] \cdot w(\tau_e).
\]
Conditioning over all choices $F_{<\tau_e^-}$ up to $\tau_e^-$, we get
\begin{align*}
    \E[\ind[\tau_e \in CC_T] \cdot w(\tau_e)]
    &\leq \E[\ind[C_e^-] \cdot \ind[\mathrm{mono}] \cdot \ind[C_e] \cdot w(\tau_e)] \\
    &= \sum \E[\ind[C_e^-] \cdot \ind[\mathrm{mono}] \cdot \ind[C_e] \cdot w(\tau_e) \mid F_{<\tau_e^-}] \Pr[F_{<\tau_e^-}] \\
    &\leq \sum \E[\ind[\mathrm{mono}] \cdot \ind[C_e] \cdot w(\tau_e) \mid F_{<\tau_e^-} , C_e^-] \Pr[F_{<\tau_e^-}] \\
    &= \sum \Pr[\ind[\mathrm{mono}] \mid F_{<\tau_e^-} , C_e^-]\E[ \ind[C_e] \cdot w(\tau_e) \mid F_{<\tau_e^-} , \mathrm{mono},C_e^-] \Pr[F_{<\tau_e^-}] \\
    &\leq \sum \frac{1}{3^{\ell^-(e)}} \E[\ind[C_e] \cdot w(\tau_e) \mid F_{<\tau_e^-} , \mathrm{mono}] \Pr[F_{<\tau_e^-}] \\
    &\leq \sum \frac{1}{3^{\ell^-(e)}} \E[ w(\tau_e) \mid F_{<\tau_e^-} , \mathrm{mono} , C_e, C_e^-] \Pr[F_{<\tau_e^-}] \\
    &=  \sum \frac{1}{3^{\ell^-(e)}} 3^{\ell^-(e)} \Pr[F_{<\tau_e^-}] = \sum \Pr[F_{<\tau_e^-}] = 1.
\end{align*}
Going from the fourth to fifth line is the deferred decision argument which uses that sampling from a clean palette is large, so
\[
    \Pr[\ind[\mathrm{mono}] \mid F_{<\tau_e^-} , C_e^-] = O\left(\frac{1}{3^{\ell^-(e)}}\right).
\]
We go from the second-to-last line to the last line by observing that $C_e$ and $C_e^-$ holding implies that $\ell(e) = \ell^-(e)$.
By taking the expectation of equation \ref{eq:cc-work-identity} and applying the above bound, we deduce that
\[
    \E[w(CC_T)] = O(T).
\]
This establishes the lemma.
\end{proof}

\paragraph{Dirty part analysis.} The first step in understanding the dirty part is to upper bound the total cost solely by the \emph{movement} cost for the dirty timestamps. To distinguish the movement cost from the total cost, we will use $w^M(-)$ instead of $w(-)$. Our goal is to prove the following lemma.
\begin{lemma}\label{lem:dirty-to-movement-cost}
    \[w(D_T) = O(\E[w^M(D_T)]).\]
\end{lemma}
We split the dirty part into the upper marked part $U_T$ and lower marked part $L_T$, so $D_T = U_T \cup L_T$. By Lemma \ref{lem:base-level-prop}, we have $\ell_i(u) \leq b_i(u) \leq \ell_i(u)+1$ for all $\tau_u^i \in L_T$. Therefore
\[
    w(L_T) = O\left(\sum_{k = 1}^\lambda 3^{\ell_i(u)} |L_T(k)|\right).
\]
Reinterpreting Lemma \ref{lem:lower-work} in terms of timestamps, and applying Lemma \ref{lem:marked-stable}, we get that
\[
    w^M(L_T) = O\left(\sum_{k = 1}^\lambda 3^{\ell_i(u)} |L_T(k)|\right).
\]
Thus $w(L_T) = O(w^M(L_T))$.

It remains to bound the upper marked part. By the definition of the cost of a dirty timestamp, we have 
\begin{equation*}
     w(U_T) = \sum_{\tau_u^i \in U_T} 3^{\max\{b_i(u),\ell_{i+1}(u)\}}.
\end{equation*}
Let $Q := {\{\tau_u^i \in U_T : \ell_{i+1}(u) > \ell_i(u)\}}$, which is the set upper marked timestamps that actually result in movement upwards. Then, due to Lemma \ref{lem:base-level-prop}, we have
\begin{equation}\label{eq:upper-marked-cost-ineq}
     w(U_T) \leq \sum_{\tau_u^i \in U_T} 3^{b_i(u)} + \sum_{\tau_u^i \in Q} 3^{\ell_{i+1}(u)}.
\end{equation}
Furthermore, interpreting Lemma \ref{lem:raise-work} in terms of timestamps, and applying Lemma \ref{lem:marked-stable}, we obtain
\begin{equation}\label{eq:upper-marked-movement-cost}
    \E[w^M(U_T)] = \Omega\left(\sum_{\tau_u^i \in Q} 3^{\ell_{i+1}(u)}\right).
\end{equation}
Our goal is to relate the first term of inequality \ref{eq:upper-marked-cost-ineq} to the quantity of \ref{eq:upper-marked-movement-cost}. To that end, partition $U_T$ into the two sets $Q^+ := \{ \tau_u^i \colon \ell_i(u) \geq b_i(u)-1\}$ and $Q^- := \{ \tau_u^i \colon \ell_i(u) < b_i(u)-1\}$.
The only way for a timestamp $\tau_u^i$ to be in $Q^-$ is for at least $3^{b_i(u)-1}$ neighbors to move above $u$ in the movement process $u$ participates in. Thus, for each $\tau_u^i \in Q^-$, there is some set of neighbors $S_u^i$ of $u$ such that $u \in N_v(1,\ell_{i+1}(v)-1)$ for $\tau_v^i \in S_u \subseteq Q^+$ where $|S_u^i| \geq 3^{b_i(u)-1}$. Let $D_u^i = \{(u,v) : \tau_v^i \in S_u^i \}$. Note that the $D_u^i$ are disjoint and $|D_u^i| = |S_u^i| = 3^{b_i(u)-1}$. Also
\[
    \bigcup_{\tau_u^i \in Q^-} D_u^i \subseteq \bigcup_{\tau_v^i \in Q} E_v^{< \ell_{i+1}(v)}, \quad \text{ so } \sum_{\tau_u^i \in Q^-} |D_u^i| \leq \sum_{\tau_v^i \in Q} |E_v^{< \ell_{i+1}(v)}|.
\]
This is because $(u,v) \in D_u^i$ implies that $i \leq \ell_{i+1}(u) < \ell_{i+1}(v)$, so $(u,v) \in E_v^{< \ell_{i+1}(v)}$. Note also that all vertices after being raised satisfy $|E_v^{< \ell_{i+1}(v)}| \leq 3^{\ell_{i+1}(v)}$. Finally, we also remark that $Q^+ \subseteq Q$ because of lemma \ref{lem:base-level-prop}.
We can then bound the first term in inequality \ref{eq:upper-marked-cost-ineq} by
\begin{align*}
     \sum_{\tau_u^i \in U_t} 3^{b_i(u)} &= \sum_{\tau_u^i \in Q^+} 3^{b_i(u)}  + \sum_{\tau_u^i \in Q^-} 3^{b_i(u)} \leq \sum_{\tau_u^i \in Q^+} 3^{\ell_{i+1}(u)} + \sum_{\tau_u^i \in Q^-} 3|D_u^i| \\
     &\leq \sum_{\tau_u^i \in Q^+} 3^{\ell_{i+1}(u)} + 3 \sum_{\tau_u^i \in Q^-} |D_u^i| \leq \sum_{\tau_u^i \in Q^+} 3^{\ell_{i+1}(u)} + 3 \sum_{\tau_v^i \in Q} |E_v^{< \ell_{i+1}(v)}| \\
     &\leq \sum_{\tau_u^i \in Q^+} 3^{\ell_{i+1}(u)} + 3 \sum_{\tau_v^i \in Q} 3^{\ell_{i+1}(v)} \leq \sum_{u \in Q} 3^{\ell_{i+1}(u)} + 3 \sum_{\tau_v^i \in Q} 3^{\ell_{i+1}(v)} \\
     &= 4\sum_{\tau_u^i \in Q} 3^{\ell_{i+1}(u)}.
\end{align*}
The above, with equations \ref{eq:upper-marked-cost-ineq} and \ref{eq:upper-marked-movement-cost}, yields
\[
    w(U_T) = O(\E[w^M(U_T)]).
\]
Putting the upper and lower marked part bounds together, we deduce Lemma \ref{lem:dirty-to-movement-cost}. The next step is to relate the movement cost to token release. Lemmas \ref{lem:adam-1} and \ref{lem:adam-2}, together with Lemma \ref{lem:marked-stable}, combine to give us the following lemma.
\begin{lemma}\label{lem:move-prop-token-release}
    The quantity $w^M(\E[D_T])$ is proportional to the amount of tokens released in performing movements.
\end{lemma}
As mentioned in Section \ref{sec:moving-phase}, every edge update injects at most $\lambda$ tokens into the system, so the total amount of released tokens is at most $\lambda T$. Combining this fact with Lemmas \ref{lem:dirty-to-movement-cost} and \ref{lem:move-prop-token-release} allows us to conclude the main theorem.